\documentclass[11pt,letter,reqno]{amsproc}

\usepackage{amsmath, amsthm, amssymb, amsfonts}
\usepackage{bbm}
\usepackage{bbold}
\usepackage[basic]{complexity}
\usepackage[pdftex,bookmarks,colorlinks]{hyperref}

\usepackage{fullpage}
\usepackage{fancyhdr} 

\newtheorem{theorem}{Theorem}
\newtheorem{lemma}[theorem]{Lemma}

\newtheorem{fact}[theorem]{Fact}
\newtheorem{definition}[theorem]{Definition}
\newtheorem{assumption}[theorem]{Assumption}

\renewcommand{\R}{\mathbb{R}}
\newcommand{\Rn}{\R^{n}}

\newcommand{\Redgevec}{\R^{E}}
\newcommand{\Rvertvec}{\R^{V}}
\newcommand{\Rev}{\R^{E \times V}}

\newcommand{\rn}{\Rn}

\newcommand{\rPos}{\R^{+}}

\newcommand{\boldVar}[1]{\mathbf{#1}} % bold variable
\newcommand{\mvar}[1]{\boldVar{#1}} % matrix variable
\newcommand{\vvar}[1]{\vec{#1}} % vector variable

\newcommand{\defeq}{\stackrel{\mathrm{\scriptscriptstyle def}}{=}}

\DeclareMathOperator*{\argmax}{arg\,max}
\DeclareMathOperator*{\argmin}{arg\,min}

\newcommand{\gradient}{\bigtriangledown}

\newcommand{\onesVec}{\vvar{\mathbb{1}}}

\newcommand{\indicVec}[1]{\onesVec_{#1}}

\newcommand{\innerProduct}[2]{\big\langle #1 , #2 \big\rangle}
\newcommand{\norm}[1]{\|#1\|}
\newcommand{\normFull}[1]{\left\|#1\right\|}
\newcommand{\normInf}[1]{\norm{#1}_{\infty}}
\newcommand{\normInfFull}[1]{\normFull{#1}_{\infty}}

\newcommand{\normTwo}[1]{\norm{#1}_{2}}

\newcommand{\normDual}[1]{\norm{#1}^*}

\newcommand{\smax}{\mathrm{smax}}

\newcommand{\dualVec}[1]{{#1}^{\#}}

\newcommand{\opt}{\mathrm{opt}}

\newcommand{\fopt}{f^*}

\newcommand{\vx}{\vvar{x}}
\newcommand{\vy}{\vvar{y}}
\newcommand{\vs}{\vvar{s}}
\newcommand{\vc}{\vvar{c}}
\newcommand{\vzero}{\vvar{0}}
\newcommand{\xopt}{\vvar{x}^*}
\newcommand{\varVec}{\vvar{x}}
\newcommand{\varVecA}{\vvar{x}}
\newcommand{\varVecB}{\vvar{y}}
\newcommand{\varMat}{\mvar{A}}

\newcommand{\abs}[1]{\left|#1\right|}

\newcommand{\capacityMatrix}{\mvar{U}}

\newcommand{\runtime}{\mathcal T}
\newcommand{\timeOf}[1]{\runtime\left(#1\right)}
\newcommand{\load}{\mathrm{load}}

\newcommand{\funLip}{L}

\newcommand{\redgevec}{\R^{E}}
\newcommand{\rvertvec}{\R^{V}}

\newcommand{\lap}{\mvar{\mathcal L}}
\newcommand{\pseudo}[1]{{#1}^\dagger}
\newcommand{\lapPseudo}{\pseudo{\lap}}

\newcommand{\incMatrix}{\mvar{B}}
\newcommand{\rMatrix}{\mvar{R}} % resistance matrix
\newcommand{\iMatrix}{\mvar{I}} % identity matrix
\newcommand{\wMatrix}{\mvar{W}}
\newcommand{\dMatrix}{\mvar{D}}

\newcommand{\conductance}{\Phi}
\newcommand{\vol}{\mathrm{vol}}
\newcommand{\cutset}{\partial}

\DeclareMathOperator*{\diag}{diag}

\newcommand{\tree}{T}
\newcommand{\treePath}[1]{P_{#1}}
\newcommand{\treePathVec}[1]{\vvar{p}_{#1}}

\newcommand{\offtreeEdgeSet}{E \setminus \tree}

\newcommand{\sampleProbVec}[1]{\vvar{p}}

\newcommand{\demands}{\vvar{\chi}}
\newcommand{\demandsEdge}[1]{\demands_{#1}}
\newcommand{\flow}{\vvar{f}}
\newcommand{\volt}{\vvar{v}}

\newcommand{\circVec}{\vvar{c}} % Variable for an arbitrary circulation

\renewcommand{\epsilon}{\varepsilon} 

\newcommand{\capacityVec}{\vvar{\mu}}
\newcommand{\capacityRatio}{U}
\newcommand{\congest}{\mathrm{cong}}

\newcommand{\mProj}{\widetilde{\mvar{P}}}
\newcommand{\mProjScaled}{\mvar{P}}

\newcommand{\mEmbed}{\mvar{M}}
\newcommand{\mRoute}{\mvar{A}}
\newcommand{\mElRout}{\mvar{A}_{\mathcal{E}}}

\newcommand{\first}{v^1}
\newcommand{\last}{v^2}

\newcommand{\im}{\mathrm{im}}

\newcommand{\lemmaref}[1]{Lemma~\ref{#1}}
\newcommand{\theoremref}[1]{Theorem~\ref{#1}}
\newcommand{\sectionref}[1]{Section~\ref{#1}}

\newcommand{\otilde}{\widetilde{O}}

\begin{document}

\title{An Almost-Linear-Time Algorithm for Approximate Max Flow in Undirected Graphs, and its Multicommodity Generalizations}

\maketitle
\begin{center}
\parbox{1.5in}{\center 
    Jonathan A.\ Kelner\\
    \texttt{\href{mailto:kelner@mit.edu}{\color{black}kelner@mit.edu}}\\
    MIT}
\parbox{1.5in}{\center 
    Yin Tat Lee\\
    \texttt{\href{mailto:yintat@mit.edu}{\color{black}yintat@mit.edu}} \\
    MIT}
\parbox{1.5in}{\center 
    Lorenzo Orecchia\\
    \texttt{\href{mailto:orecchia@mit.edu}{\color{black}orecchia@mit.edu}} \\
    MIT}
\parbox{1.5in}{\center 
    Aaron Sidford\\
\texttt{\href{mailto:sidford@mit.edu}{\color{black}sidford@mit.edu}} \\
    MIT}
\end{center}
{\ }\\{\ }\\

\date{}

\thispagestyle{empty}

\begin{abstract}

In this paper, we introduce a new framework for approximately solving flow problems in capacitated, undirected graphs and apply it to provide asymptotically faster algorithms for the maximum $s$-$t$
flow and maximum concurrent multicommodity flow problems.
For graphs with $n$ vertices and $m$ edges, it allows us to find an $\epsilon$-approximate maximum $s$-$t$ flow in time ${O}(m^{1+o(1)}\epsilon^{-2})$, improving on the previous best bound of $\otilde(mn^{1/3}\mathrm{poly(1/\epsilon)})$.
Applying the same framework in the multicommodity setting
solves a maximum concurrent multicommodity flow problem with $k$ commodities
in ${O}(m^{1+o(1)}\epsilon^{-2}k^2)$ time, improving on the existing bound of $\otilde(m^{4/3}\mathrm{poly}(k,\epsilon^{-1}))$.\\

Our algorithms utilize several new technical tools that we believe may be of independent interest:\\

 \begin{itemize} \item We give a non-Euclidean generalization of gradient descent and provide bounds on its performance.  Using this, we show how to 
 reduce approximate maximum flow and maximum concurrent flow to oblivious routing. \\
  \item We define and provide an efficient construction of a new type of \emph{flow sparsifier}.  Previous sparsifier constructions approximately preserved the size of cuts and, by duality, the \emph{value} of the maximum flows as well.  However, they did not provide any direct way 
to route flows in the sparsifier $G'$ back in the original graph $G$, leading to a longstanding gap between the efficacy of sparsification on flow and cut problems.
We ameliorate this by constructing a sparsifier $G'$ that can be embedded (very efficiently) into $G$ with low congestion, allowing one to transfer flows from $G'$ back to $G$.\\
    \item We give the first almost-linear-time construction of an $O(m^{o(1)})$-competitive oblivious routing scheme. No previous such algorithm ran in time better than $\widetilde{\Omega}(mn)$.
    By reducing the running time to almost-linear,   our work provides a powerful new primitive for constructing very fast graph algorithms.\\
%\end{list}
\end{itemize}

We also note that independently Jonah Sherman produced an almost linear time algorithm for maximum flow and we thank him for coordinating submissions.

\end{abstract}

% Make sure abstract is on its own page and that numbering starts at 1
\newpage
\setcounter{page}{1}

\section{Introduction}

Given a graph $G=(V,E)$ in which each edge $e\in E$ is assigned a nonnegative capacity $\capacityVec_e$, the \textbf{maximum $s$-$t$ flow problem} asks us to find a flow $\flow$ that routes as much flow as possible from a source vertex $s$ to a sink vertex $t$ while sending at most $\capacityVec_e$ units of flow over each edge $e$.
Its generalization, the \textbf{maximum concurrent multicommodity flow problem}, supplies $k$ source-sink pairs $(s_i,t_i)$ and asks for the maximum $\alpha$ such that we may simultaneously route $\alpha$ units of flow between each source-sink pair. 
That is, it asks us to find flows $\flow_1, \dots, \flow_k$ (which we think of as corresponding to $k$ different commodities) such that $\flow_i$ sends $\alpha$ units of flow from $s_i$ to $t_i$, and $\sum_{i} |\flow_i(e)|\leq \capacityVec_e$ for all $e \in E$.

These problems lie at the core of graph algorithms and combinatorial optimization, and  they have been extensively studied  over the past 60 years~\cite{SchrijverA,AhujaBook}.  They have found a wide range of theoretical and practical applications~\cite{AhujaApps}, and they are widely used as key subroutines in other algorithms (see~\cite{AroraHK05,ShermanMaxflow}).

In this paper, we introduce a                                            
new framework for approximately solving flow problems in capacitated, undirected   graphs and apply it to provide asymptotically faster algorithms for the maximum $s$-$t$
flow and maximum concurrent multicommodity flow problems.
For graphs with $n$ vertices and $m$ edges, it allows us to find an $\epsilon$-approximately maximum $s$-$t$ flows in time $O(m^{1+o(1)}\epsilon^{-2})$, improving on the previous best bound of $\otilde(mn^{1/3}\mathrm{poly(1/\epsilon)})$\cite{CKMST}.
Applying the same framework 
in the multicommodity setting
solves a maximum concurrent multicommodity flow problem with $k$ commodities
in $O(m^{1+o(1)}\epsilon^{-2}k^2)$ time, improving on the existing bound of $\otilde(m^{4/3}\mathrm{poly}(k,\epsilon^{-1}))$\cite{KelnerMillerPeng}.

We believe that both our general framework and several of the pieces necessary for its present instantiation are of independent interest, and we hope that they will find other applications.  
These include:
\begin{itemize}
  \item a non-Euclidean generalization of gradient descent, bounds on its performance, 
 and a way to use this to reduce approximate maximum flow and maximum concurrent flow to oblivious routing; 
  \item the definition and efficient construction of \emph{flow sparsifiers}; and
  \item the construction of a new oblivious routing scheme that can be implemented extremely efficiently.\end{itemize}
  
We have aimed to make our algorithm fairly modular and have thus occasionally worked in slightly more generality than is strictly necessary for the problem at hand.  This has slightly increased the length of the exposition, but we believe that it clarifies the high-level structure of the argument, and it will hopefully facilitate the application of these tools in other settings.

\subsection{Related Work}

For the first several decades of its study, the fastest algorithms for the maximum flow problem were essentially all deterministic algorithms based on combinatorial techniques, such as augmenting paths, blocking flows, preflows, and the push-relabel method.  These culminated in the work of Goldberg and Rao~\cite{GoldbergRao}, which computes exact maximum flows in time
$O(\min(n^{2/3}, m^{1/2}) \log(n^2/m) \log U)$ on graphs with edge weights in $\{0,\dots, U\}$. We refer the reader to~\cite{GoldbergRao} for a survey of these results.

More recently, a collection of new techniques based on randomization, spectral graph theory and numerical linear algebra, graph decompositions and embeddings, and iterative methods for convex optimization have  emerged.  These have  allowed researchers to provide better provable algorithms for a wide range of flow and cut problems, particularly when one aims to obtain approximately optimal solutions on undirected graphs.

Our algorithm draws extensively on the intellectual heritage established by these works. 
In this section, we will briefly review 
some of the previous advances that inform our algorithm.
We do not give a comprehensive review of the literature, but instead aim to provide a high-level view of the main tools
 that motivated the present work, along with the limitations of these tools that had to be overcome.
 For simplicity of exposition, we primarily focus on the maximum $s$-$t$ flow problem for the remainder of the introduction.
                                                                                                                                                        \paragraph{\bf Sparsification}
In~\cite{BenczurK96}, Benczur and Karger showed how to efficiently approximate any graph $G$ with a sparse graph $G'$ on the same vertex set. To do this, they compute a carefully chosen probability $p_e$ for each $e\in E$, sample each edge $e$ with probability $p_e$, and include $e$ in $G'$ with its weight increased by a factor of $1/p_e$ if it is sampled.
Using this,
they obtain, in nearly linear time, a graph $G'$ with $O(n\log n / \epsilon^2)$ edges such that the total weight of the edges crossing any cut in $G'$ is within a multiplicative factor of $1\pm \epsilon$ of the weight crossing the corresponding cut in $G$.    In particular, the Max-Flow Min-Cut Theorem implies that the value of the maximum flow on $G'$ is within a factor of $1\pm \epsilon$ of that of $G$.

This is an extremely effective tool for approximately solving cut problems on a dense graph $G$, since one can simply solve the corresponding problem on the sparsified graph $G'$.  However, while this means that one can approximately compute the \emph{value} of the maximum $s$-$t$ flow on $G$ by solving the problem on $G'$,  it is not known how to use the maximum $s$-$t$ flow on $G'$  to obtain 
an actual approximately maximum flow on $G$.  Intuitively, this is because the weights of edges included in $G'$ are larger than they were in $G$, and the sampling argument does not provide any guidance about how to route flows over these edges in the original graph $G$.  

\paragraph{\bf Iterative algorithms based on linear systems and electrical flows}  In 2010, Christiano \emph{et al.}\cite{CKMST}
described a new linear algebraic approach to the problem that found $\epsilon$-approximately maximum $s$-$t$ flows in time $\otilde(mn^{1/3}\mathrm{poly(1/\epsilon)})$.  
They treated the edges of $G$ as electrical resistors and then computed the  \emph{electrical flow} that would result from sending electrical current  from $s$ to $t$ in the corresponding circuit.  
They showed that these flows can be computed in nearly-linear time using fast Laplacian linear system solvers~\cite{KoutisMP10,Koutis:2011:NLN:2082752.2082901,koszSolver,YinTatAaron}, which we further discuss below.
The electrical flow obeys the flow conservation constraints, but it could violate the capacity constraints.  They then adjusted the resistances of edges to penalize the edges that were flowing too much current and repeated the process. 
Kelner, Miller, and Peng~\cite{KelnerMillerPeng} later showed how to use more general objects that they called \emph{quadratically coupled flows} to use a similar approach to solve the maximum concurrent multicommodity flow problem in time 
$\otilde(m^{4/3}\mathrm{poly}(k,1/\epsilon))$.

Following this, Lee, Rao, and Srivastava~\cite{Lee2013} proposed another iterative algorithm that uses electrical flows, but in a way that was substantially different than in~\cite{CKMST}.  Instead of adjusting the resistances of the edges in each iteration to correct overflowing edges, they keep the resistances the same but compute a new electrical flow to reroute the excess current.  They explain how to interpret this as gradient descent in a certain space, from which a standard analysis would give an algorithm that runs in time $\otilde(m^{3/2}\mathrm{poly}(1/\epsilon))$.  By replacing the standard gradient descent step with Nesterov's accelerated gradient descent method~\cite{nesterov2003introductory} and using a regularizer to make the penalty function smoother, they obtain an algorithm that runs in time $\otilde(mn^{1/3}\mathrm{poly}(1/\epsilon))$ in unweighted graphs.

In all of these algorithms, the superlinear running times arise from an intrinsic $\Theta(\sqrt{m})$ 
factor introduced by using electrical flows, which minimize an $\ell_2$ objective function, to approximate the maximum congestion, which is an $\ell_\infty$ quantity.
 
\paragraph{\bf Fast solvers for Laplacian linear systems}
In their breakthrough paper~\cite{SpielmanTeng04}, Spielman and Teng showed how to solve Laplacian systems in nearly-linear time.  (This was later sped up and simplified by Koutis, Miller, and Peng~\cite{KoutisMP10,Koutis:2011:NLN:2082752.2082901} and Kelner, Orecchia, Sidford, and Zhu~\cite{koszSolver}.)
Their algorithm worked by showing how to approximate the Laplacian $\lap_G$ of a graph $G$ with the Laplacian $\lap_H$ of a much simpler graph $H$ such that one could use the ability to solve linear systems in $\lap_H$ to accelerate the solution of a linear system in $\lap_G$.  They then applied this recursively to solve the linear systems in $\lap_H$.
In addition to providing the electrical flow primitive used by the algorithms described above, 
 the structure of their recursive sequence of graph simplifications provides the motivating framework for much of the technical content 
 of our oblivious routing construction.

\paragraph{\bf Oblivious routing}
In an \emph{oblivious routing scheme}, one specifies a linear operator 
taking any demand vector to a flow routing these demands over the edges of $G$.  
Given a collection of demand vectors, one can produce a multicommodity flow meeting these demands by routing each demand vector using this pre-specified operator, independently of the others.  The competitive ratio of such an operator is the worst possible ratio between the congestion incurred by a set of demands in this scheme and the congestion of the best multicommodity flow routing these demands.  

In~\cite{Racke:2008:OHD:1374376.1374415}, R\"acke showed how to construct an oblivious routing scheme with a competitive ratio of $O(\log n)$.
His construction worked by providing a probability distribution over trees $T_i$ such that $G$ embeds into each $T_i$ with congestion at most 1, and such that the corresponding convex combination of trees embeds into $G$ with congestion $O(\log n)$.  In a sense, one can view this as showing how to approximate $G$ by a probability distribution over trees.
Using this, he was able to show how to obtain polylogarithmic approximations for a variety of cut and flow problems, given only the ability to solve these problems on trees.

We note that such an oblivious routing scheme clearly yields a logarithmic approximation to the maximum flow and maximum concurrent multicommodity flow problems.  
However, R\"acke's construction took time substantially superlinear time, making it too slow to be useful for computing approximately maximum flows.  Furthermore, it only gives a logarithmic approximation, and it is not clear how to use this a small number of times to reduce the error to a multiplicative  $\epsilon$.

In a later paper~\cite{Madry10}, Madry applied a recursive technique similar to the one employed by Spielman and Teng in their Laplacian solver to accelerate many of the applications of R\"acke's construction at the cost of a worse approximation ratio. Using this, he obtained almost-linear-time polylogarithmic approximation algorithms for a wide variety of cut problems.  

Unfortunately, his algorithm made extensive use of sparsification, which, for the previously mentioned reasons, made it unable to solve the corresponding flow problems.  This meant that, while it could use flow-cut duality to find a polylogarithmic approximation of the value of a maximum flow, it could not construct a corresponding flow or repeatedly apply such a procedure a small number of times to decrease the error to a multiplicative $\epsilon$.

In simultaneous, independent work~\cite{ShermanMaxflow}, Jonah Sherman used somewhat different techniques to find another almost-linear-time algorithm for the (single-commodity) maximum flow problem. 
His approach is essentially dual to ours: 
Our algorithm maintains a flow that routes the given demands throughout its execution and iteratively works to improve its congestion.  Our main technical tools thus consist of efficient methods for finding ways to route flow in the graph while maintaining flow conservation.
Sherman, on the other hand, maintains a flow that does \emph{not} route the given demands, along with a bound on the congestion required to route the excess flow at the vertices.  He then uses this to iteratively work towards achieving flow conservation.  (In a sense, our algorithm is more in the spirit of augmenting paths, whereas his is more like preflow-push.) As such, his main technical tools are efficient methods for producing dual objects that give congestion bounds.  
Objects meeting many of his requirements were given in the work of Madry~\cite{Madry10}  (whereas there were no previous constructions of flow-based analogues, requiring us to start from scratch); leveraging these allows him to avoid some of the technical complexity required by our approach.
 We believe that these paper nicely complement each other, and we enthusiastically refer the reader to Sherman's paper.

\subsection{Our Approach}
In this section, we give a high-level description of how we overcome the obstacles described in the previous section.  For simplicity, we suppose for the remainder of this introduction that all edges have capacity 1. 

The problem is thus to send as many units of flow as possible from $s$ to $t$ without sending more than one unit over any edge.  It will be more convenient for us to work with an equivalent congestion minimization problem, where we try to find the unit  $s$-$t$ flow $\flow$ 
(i.e., a flow sending one unit from $s$ to $t$)
that minimizes $\|f\|_\infty= \max_e |\flow_e|$.  
If we begin with some initial unit $s$-$t$ flow $\flow_{0}$, the goal will be thus be to find the  circulation $\circVec$ to add to $\flow_0$ that minimizes $\|\flow_0+\circVec\|_\infty$.

We give an iterative algorithm to approximately find such a $\circVec$.  There are $2^{O(\sqrt{\log n \log \log n})}/\epsilon^2$ iterations, each of which adds a circulation to the present flow and runs in $m\cdot 2^{O(\sqrt{\log n \log \log n})}$ time.  
Constructing this scheme consists of two main parts: an iterative scheme that reduces the problem to the construction of a projection matrix with certain properties; and the construction of such an operator.

\subsubsection*{The iterative scheme: Non-Euclidean gradient descent}

The simplest way to improve the flow would be to just perform gradient descent on the maximum congestion of an edge.  There are two problems with this:

The first problem is that gradient descent depends on having a smoothly varying gradient, but the infinity norm is very far from smooth. This is easily remedied by a standard technique: we replace the infinity norm with a smoother ``soft max'' function.
Doing this would lead to an update that would be a linear projection onto the space of circulations.  This could be computed using an electrical flow, and the resulting algorithm would be very similar to the unaccelerated gradient descent algorithm in~\cite{Lee2013}.  

The more serious problem is the one discussed in the previous section: the difference between $\ell_2$ and $\ell_\infty$.
Gradient steps choose a direction by optimizing a local approximation of the objective function over a sphere, whereas the $\ell_\infty$ constraint  asks us to optimize over a cube.  The difference between the size of the largest sphere inside a cube and the smallest sphere containing it gives rise to an inherent $O(\sqrt{m})$ in the number of iterations, unless one can somehow exploit additional structure in the problem.

To deal with this, we introduce and analyze a non-Euclidean variant of gradient descent that operates with respect to an arbitrary norm.\footnote{This idea and analysis seems to be implicit in other work, e.g.,~\cite{nesterov2010efficiency}  However, we could not find a clean statement like the one we need in the literature, and we have not seen it previously applied in similar settings.  We believe that it will find further applications, so we state it in fairly general terms before specializing to what we need for flow problems.}  Rather than choosing the direction by optimizing a local linearization of the objective function over the sphere, it performs an optimization over the unit ball in the given norm.  By taking this norm to be  $\ell_\infty$ instead of $\ell_2$, we are able to obtain a much smaller bound on the number of iterations, albeit at the expense of having to solve a nonlinear minimization problem at every step.
The number of iterations required by the gradient descent method depends on how quickly the gradient can change over balls in the norm we are using, which we express in terms of the Lipschitz constant of the gradient in the chosen norm.

To apply this to our problem, we write flows meeting our demands as $\flow_0 +\circVec$, as described above.  We then need a parametrization of the space of circulations so that the objective function (after being smoothed using soft max) has a good bound on its Lipschitz constant.  
Similarly to what occurs in~\cite{koszSolver}, this comes down to finding a good linear representation of the space of circulations, which we show amounts in the present setting to finding a matrix that projects into the space of circulations while meetings certain norm bounds. 

\subsubsection*{Constructing a projection matrix}
This reduces our problem to the construction of such a projection matrix.
A simple calculation shows that any linear oblivious routing scheme $A$ with a good competitive ratio gives rise to a projection matrix with the desired properties, and thus leads to an iterative algorithm that converges in a small number of iterations.  Each of these iterations performs a matrix-vector multiplication with both $\mvar{A}$ and $\mvar{A}^T$.

Intuitively, this is letting us replace the electrical flows used in previous algorithms with the flows given by an oblivious routing scheme.  
Since the oblivious routing scheme was constructed to meet $\ell_{\infty}$ guarantees, while the electrical flow could only obtain such guarantees by relating $\ell_2$ to $\ell_\infty$, it is quite reasonable that we should expect this to lead to a better iterative algorithm.

However, the computation involved  in existing oblivious routing schemes is not fast enough to be used in this setting.  Our task thus becomes constructing an oblivious routing scheme that we can compute and work with very efficiently.  We do this with a recursive construction that reduces oblivious routing in a graph to oblivious routing in various successively simpler graphs.

To this end, we show that if $G$ can be embedded with low congestion into $H$ (existentially), and $H$ can be embedded with low congestion into $G$ \emph{efficiently}, one can use an oblivious routing on $H$ to obtain an oblivious routing on $G$.
The crucial difference between the simplification operations we perform here and those in previous papers (e.g., in the work of Benczur-Karger~\cite{BenczurK96} and Madry~\cite{Madry10}) is that ours are accompanied by such embeddings, which enables us to transfer flows from the simpler graphs to the more complicated ones. 

We construct our routing scheme by recursively composing two types of reductions, each of which we show how to implement without incurring a large increase in the competitive ratio:
\begin{list}{$\bullet$}{\leftmargin=1em \itemindent=0em}
\item {\bf Vertex elimination} This  shows how to efficiently reduce oblivious routing on a graph $G=(V,E)$ to  routing on $t$ graphs with roughly $\otilde(|E|/t)$ vertices.    To do this, we show how to efficiently embed $G$ into  $t$ simpler graphs, each consisting of a tree plus a subgraph supported on roughly $\otilde{(|E|/t)}$ vertices.  This follows easily from a careful reading of Madry's paper~\cite{Madry10}.  We then show that routing on such a graph can be reduced to routing on a graph with at most $\otilde{(|E|/t)}$ vertices by collapsing paths and eliminating leaves.
\vspace{0.5em}
\item {\bf Flow sparsification} This allows us to efficiently reduce oblivious routing on an arbitrary graph to oblivious routing on a graph with $\otilde(|V|)$ edges, which we call a flow sparsifier.

To  construct flow sparsifiers, we use local partitioning to decompose the graph into well-connected clusters that contain many of the original edges.  
(These clusters are not quite expanders, but they are contained in graphs with good expansion in a manner that is sufficient for our purposes.)  We then sparsify these clusters using standard techniques and then show that we can embed the sparse graph back into the original graph using electrical flows. If the graph was originally dense, this results in a sparser graph, and we can recurse on the result.
While the implementation of these steps is somewhat different, the outline of this construction parallels Spielman and Teng's approach to the construction of spectral sparsifiers~\cite{SpielmanTeng04, STspectralSparse}.
\end{list}

Combining these two reductions recursively yields an efficient oblivious routing scheme, and thus an algorithm for the maximum flow problem.

Finally, we show that the same framework can be applied to the maximum concurrent  multicommodity flow problem.  While the norm and regularization change, the structure of the argument and the construction of the oblivious routing scheme go through without requiring substantial modification.

\section{Preliminaries}

\textbf{General Notation:} We typically use $\vx$ to denote a vector and $\mvar{A}$ to denote a matrix. For $\vx \in \R^n$, we let $\abs{\vx} \in \R^{n}$ denote the vector such that $\forall i$, $\abs{\vx}_i \defeq \abs{\vx_i}$. For a matrix $\mvar{A} \in \R^{n \times m}$, we let $\abs{\mvar{A}}$ denote the matrix such that $\forall i, j$ we have $\mvar{|A|}_{ij} \defeq \abs{\mvar{A}_{ij}}$. We let $\onesVec$ denote the all ones vector and $\onesVec_i$ denote the vector that is one at position $i$ and 0 elsewhere. We let $\iMatrix$ be the identity matrix and  $\iMatrix_{a \rightarrow b} \in \R^{b \times a}$ denote the matrix such that for all $i \leq \min\{a, b\}$ we have $\iMatrix_{ii} = 1$ and $\iMatrix_{ij} = 0$ otherwise.

\textbf{Graphs:} Throughout this paper we let $G = (V, E, \capacityVec)$ denote an undirected capacitated graph with $n = |V|$ vertices, $m = |E|$ edges, and non-negative capacities $\capacityVec \in \redgevec$. We let $w_e \geq 0$ denote the weight of an edge and let $r_e \defeq 1/w_e$ denote the resistance of an edge. Here we make no connection between $\mu_e$ and $r_e$; we fix their relationship later. While all graphs in this paper are undirected, we frequently assume an arbitrary orientation of the edges to clarify the meaning of vectors $\flow \in \redgevec$.

\textbf{Fundamental Matrices:} Let $\capacityMatrix, \wMatrix, \rMatrix \in \R^{E \times E}$ denote the diagonal matrices associated with the capacities, the weights, and the resistances respectively. Let $\incMatrix \in \R^{E \times V}$ denote the graphs incidence matrix where for all $e = (a, b) \in E$ we have $\incMatrix^T \indicVec{e} = \indicVec{a} - \indicVec{b}$. Let $\lap \in \R^{V \times V}$ denote the combinatorial graph Laplacian, i.e. $\lap \defeq \incMatrix^T \rMatrix^{-1} \incMatrix$.

\textbf{Sizes:} For all $a \in V$ we let $d_a \defeq \sum_{\{a, b\}} w_{a,b}$ denote the \emph{(weighted) degree} of vertex $a$ and we let $\deg(a) \defeq |\{e \in E ~ | ~ e = \{a, b\} \text{ for some } b \in V\}|$ denote its \emph{(combinatorial) degree}. We let $\dMatrix \in \R^{V \times V}$ be the diagonal matrix where $\dMatrix_{a, a} = d_a$. Furthermore, for any vertex subset $S \subseteq V$ we define its \emph{volume} by
$
\vol(S) \defeq \sum_{a \in V} d_a
$
.

\textbf{Cuts:} For any vertex subset $S \subseteq V$ we denote the cut induced by $S$ by the edge subset
\[
\cutset(S)
\defeq \{e \in E ~ | ~ e \notin S \text{ and } e \notin E \setminus S\}
\]
and we denote the cost of $F \subseteq E$ by $w(F) \defeq \sum_{e \in F}w_e$. We denote the \emph{conductance} of $S \subseteq V$ by
\[
\conductance(S) \defeq \frac{w(\cutset(S))}{\min\{\vol(S), \vol(V - S)\}}
\]
and we denote the conductance of a graph by
\[
\conductance(G) \defeq \min_{S \subseteq V ~ : S \notin \{\emptyset, V\}} \phi(S)
\]

\textbf{Subgraphs:} For a graph $G = (V, E)$ and a vertex subset $S \subseteq V$ let $G(S)$ denote the subgraph of $G$ consisting of vertex set $S$ and all the edges of $E$ with both endpoints in $S$, i.e. $\{(a, b) \in E ~ | ~ a, b \in S\}$. When we we consider a term such as $\vol$ or $\conductance$ and we wish to make the graph such that this is respect to clear we use subscripts. For example $\vol_{G(S)}(A)$ denotes the volume of vertex set $A$ in the subgraph of $G$ induced by $S$.

\textbf{Congestion:} Thinking of edge vectors, $\flow \in \redgevec$, as flows we let the \emph{congestion\footnote{Note that here and in the rest of the paper we will focus our analysis with congestion with respect to the norm $\normInf{\cdot}$ and we will look at oblivious routing strategies that are competitive with respect to this norm. However, many of the results present are easily generalizable to other norms. These generalizations are outside the scope of this paper.} of $\flow$} be given by $\congest(\flow) \defeq \normInf{\capacityMatrix^{-1} \flow}$. For any collection of flows $\{\flow_i\} = \{\flow_1, \ldots, \flow_k\}$ we overload notation and let their \emph{total congestion} be given by
\[
\congest(\{\flow_i\})
\defeq 
\normInfFull{\capacityMatrix^{-1} \sum_i |\flow_i|}
\]

\textbf{Demands and Multicommodity Flow:}
We call a vector $\demands \in \rvertvec$ a \emph{demand vector} if itis the case that $\sum_{a \in V} \demands(a) = 0$ and we say $\flow \in \redgevec$ meets the demands if $\incMatrix^T \flow = \demands$. Given a set of demands $D = \{\demands_1, \ldots, \demands_k\}$, i.e. $\forall i \in [k], \sum_{a \in V} \demands_i(a) = 0$, we denote the optimal low congestion routing of these demands as follows
\[
\opt(D)
\defeq \min_{\{\flow_i\} \in \redgevec ~ : ~ \{\incMatrix^T \flow_i\} = \{\demands_i\}}
\congest(\{\flow_i\})
\]
We call a set of flows $\{\flow_i\}$ that meet demands $\{\demands_i\}$, i.e. $\forall i, \incMatrix^T \flow_i = \demands_i$, a \emph{multicommodity flow} meeting the demands.

\textbf{Operator Norm:} Let $\norm{\cdot}$ be a family of norms applicable to $\Rn$ for any $n$. We define this norms' \emph{induced norm} or \emph{operator norm} on the set of of $m \times n$ matrices by
\[
\forall \mvar{A} \in \R^{n \times m}
\enspace : \enspace
\norm{\mvar{A}} \defeq \max_{\vx \in \Rn} \frac{\norm{\mvar{A} \vx}}{\norm{\vx}}
\]

\textbf{Running Time:} For matrix $\mvar{A},$ we let $\timeOf{\mvar{A}}$ denote the maximum amount of time needed to apply $\mvar{A}$ or $\mvar{A}^T$ to a vector.

\section{Solving Max-Flow Using a Circulation Projection}

\subsection{Gradient Descent}

In this section, we discuss the gradient descent method for general
norms. Let $\norm{\cdot}:\R^{n}\rightarrow\R$ be an arbitrary norm
on $\R^{n}$ and recall that the gradient of $f$ at $\vx$ is defined
to be the vector $\nabla f(\vx) \in \Rn$ such that 
\begin{equation}
f(\vy)=f(\vx)+\left\langle \nabla f(\vx),\vy-\vx\right\rangle +o(\norm{\vy-\vx}).\label{eq:taylor}
\end{equation}
The gradient descent method is a greedy minimization method that updates
the current vector, $\vx$, using the direction which minimizes $\left\langle \nabla f(\vx),\vy-\vx\right\rangle $.
To analyze this method's performance, we need a tool to compare the improvement
$\left\langle \nabla f(\vx),\vy-\vx\right\rangle $ with the step
size, $\norm{\vy-\vx}$, and the quantity, $\norm{\nabla f(\vx)}$. For $\ell_2$
norm, this can be done by Cauchy Schwarz inequality and in general, we
can define a new norm for $\nabla f(\vx)$ to make this happens. We
call this the \emph{dual norm} $\normDual{\cdot}$ defined as follows
\[
\normDual{\varVecA}\defeq\max_{\varVecB\in\Rn~:~\norm{\varVecB}\leq1}\innerProduct{\varVecB}{\varVecA}.
\]
Fact \ref{claim:CS_inq} shows that this definition indeed yields that
$\innerProduct{\vy}{\vx}\leq\normDual{\vy}\norm{\vx}$. Next, we define
the fastest increasing direction $\dualVec x$, which is an arbitrary
point satisfying the following 
\[
\dualVec{\vx}\defeq\argmax_{\vs\in\R}\innerProduct{\vx}{\vs}-\frac{1}{2}\norm{\vs}^{2}.
\]
In the appendix, we provide some facts about $\normDual{\cdot}$
and $\dualVec{\vx}$ that we will use in this section. Using the notations defined, the gradient descent
method simply produces a sequence of $\vx_{k}$ such that 
\[
\vx_{k+1}:=\vx_{k}-t_{k}\dualVec{(\gradient f(\vx_{k}))}
\]
where $t_{k}$ is some chosen step size for iteration $k$. To determine what these step sizes should be we need some information about the smoothness of the function, in
particular, the magnitude of the second order term in (\ref{eq:taylor}).
The natural notion of smoothness for gradient descent is the Lipschitz
constant of the gradient of $f$, that is the smallest constant $L$
such that 
\[
\forall\vx,\vy\in\rn\enspace:\enspace\normDual{\gradient f(\vx)-\gradient f(\vy)}\leq\funLip\cdot\norm{\vx-\vy}.
\]
In the appendix we provide an equivalent definition and a way to compute $L$, which is useful later. 

Let $X^{*}\subseteq\rn$ denote the set of optimal solutions to the unconstrained
minimization problem $\min_{x\in\R^{n}}f$ and let $\fopt$ denote
the optimal value of this minimization problem, i.e. 
\[
\forall\vx\in X^{*}~:~f(\vx)=\fopt=\min_{\vy\in\R^{n}}f(y)\enspace\text{ and }\enspace\forall\vx\notin X^{*}~:~f(\vx)>\fopt
\]
We assume that $X^{*}$ is non-empty. Now, we are ready to estimate
the convergence rate of the gradient descent method. 

\begin{theorem}[Gradient Descent] \label{thm:gradient_descent}Let
$f:\R^{n}\rightarrow\R$ be a convex continuously differentiable function
and let $L$ be the Lipschitz constant of $\nabla f$. For initial point $\vx_0 \in \rn$ we define a sequence of $\vx_{k}$
by the update rule 
\[
\vx_{k+1}:=\vx_{k}-\frac{1}{\funLip}\dualVec{(\gradient f(\vx_{k}))}
\]
For all $k\geq0$, we have 
\[
f(\vx_{k})-\fopt\leq\frac{2\cdot L\cdot R^{2}}{k+4}\enspace\text{ where }\enspace R\defeq\max_{\vx\in\rn:f(\vx)\leq f(\vx_{0})}\min_{\xopt\in X^{*}}\norm{\vx-\xopt}.
\]
\end{theorem}

\begin{proof}\footnote{The structure of this specific proof was modeled after a proof in \cite{nesterov2010efficiency} for a slightly different problem.} By the Lipschitz continuity of the gradient
of $f$ and \lemmaref{lem:lipeq} we have
\[
f(\vx_{k+1})\leq f(\vx_{k})-\frac{1}{2\funLip}\left(\normDual{\gradient f(\vx_{k})}\right)^{2}.
\]
Furthermore, by the convexity of $f$, we know that 
\[
\forall\vx,\vy\in\rn\enspace:\enspace f(\vy)\geq f(\vx)+\innerProduct{\gradient f(\vx)}{\vy-\vx}.
\]
Using this and the fact that $f(\vx_{k})$ decreases monotonically
with $k$, we get 
\[
f(\vx_{k})-\fopt\leq\min_{\xopt\in X^{*}}\innerProduct{\gradient f(\vx_{k})}{\vx_{k}-\xopt}\leq\min_{\xopt\in X^{*}}\normDual{\gradient f(\vx_{k})}\norm{\vx_{k}-\xopt}\leq R\normDual{\gradient f(\vx_{k})}.
\]
Therefore, letting $\phi_{k}\defeq f(\vx_{k})-f^{*}$, we have 
\[
\phi_{k}-\phi_{k+1}\geq\frac{1}{2L}(\normDual{\gradient f(\vx_{k})})^{2}\geq\frac{\phi_{k}^{2}}{2\cdot\funLip\cdot R^{2}}.
\]
Furthermore, since $\phi_{k}\geq\phi_{k+1}$, we have 
\[
\frac{1}{\phi_{k+1}}-\frac{1}{\phi_{k}}=\frac{\phi_{k}-\phi_{k+1}}{\phi_{k}\phi_{k+1}}\geq\frac{\phi_{k}-\phi_{k+1}}{\phi_{k}^{2}}\geq\frac{1}{2\cdot\funLip\cdot R^{2}}.
\]
So, by induction, we have that 
\[
\frac{1}{\phi_{k}}-\frac{1}{\phi_{0}}\geq\frac{k}{2\cdot\funLip\cdot R^{2}}.
\]
Now, note that since $\gradient f(\xopt)=0$, we have that 
\[
f(\vx_{0})\leq f(\xopt)+\innerProduct{\gradient f(\xopt)}{\vx_{0}-\xopt}+\frac{\funLip}{2}\norm{\vx_{0}-\xopt}^{2}\leq f(\xopt)+\frac{\funLip}{2}R^{2}.
\]
So, we have that $\phi_{0}\leq\frac{\funLip}{2}R^{2}$ and putting
this all together yields that 
\[
\frac{1}{\phi_{k}}\geq\frac{1}{\phi_{0}}+\frac{k}{2\cdot\funLip\cdot R^{2}}\geq\frac{4}{2\cdot\funLip\cdot R^{2}}+\frac{k}{2\cdot\funLip\cdot R^{2}}.
\]
\end{proof}

\subsection{Maximum Flow Formulation \label{sub:MaxFlow_formulation}}

For an arbitrary set of demands $\demands\in\rvertvec$
we wish to solve the following \emph{maximum flow} problem
\[
\max_{\alpha\in R,\flow\in\R^{E}}\alpha
\qquad\text{subject to}\qquad\incMatrix^{T}\flow=\alpha\demands\enspace\text{and}\enspace\normInf{\mvar U^{-1}\flow}\leq1.
\]
Equivalently, we want to compute a \emph{minimum congestion flow}
\[
\min_{\flow\in\redgevec~:~\incMatrix^{T}\flow=\demands}\normInf{\capacityMatrix^{-1}\flow}.
\]
where we call $\normInf{\capacityMatrix^{-1}\flow}$ \emph{the congestion}
of $\flow$.

Letting $\flow_{0}\in\redgevec$ be some initial \emph{feasible flow}, i.e.
$\incMatrix^{T}\flow_{0}\defeq\demands$, we write the problem equivalently
as 
\[
\min_{\circVec\in\redgevec~:~\incMatrix^{T}\circVec=0}\normInf{\capacityMatrix^{-1}(\flow_{0}+\circVec)}
\]
where the output flow is $\flow=\flow_{0}+\circVec$. Although the
gradient descent method is applicable to constrained optimization
problems and has a similar convergence guarantee, the sub-problem
involved in each iteration is a constrained optimization problem,
which is quite complicated in this case. Since the domain is a linear
subspace, the constraints can be avoided by projecting the variables
onto this subspace. 

Formally, we define a circulation projection matrix as follows.

\begin{definition}\label{def:circ_proj} A matrix $\mvar{\tilde{P}}\in\R^{E\times E}$ is
a \emph{circulation projection matrix} if it is a projection matrix onto
the circulation space, i.e. it satisfies the following 
\begin{itemize}
\item $\forall\vx\in\redgevec$ we have $\incMatrix^{T}\mvar{\tilde{P}}\vx= \vzero$.
\item $\forall\vx\in\redgevec$ with $\incMatrix^{T}\vx = \vzero$ we have $\mvar P\vx=\vx$.
\end{itemize}
\end{definition}

Then, the problem becomes
\[
\min_{\circVec\in\redgevec}\normInf{\capacityMatrix^{-1}(\flow_{0}+\mvar{\tilde{P}}\circVec)}.
\]
Applying gradient descent on this problem is similar to applying projected
gradient method on the original problem. But, instead of using the
orthogonal projection that is not suitable for $\normInf{\cdot}$, we
will pick a better projection matrix.

Applying the change of basis $\varVec=\capacityMatrix^{-1}\circVec$
and letting $\vvar{\alpha_{0}}=\capacityMatrix^{-1}\flow_{0}$ and
$\mvar P=\mvar U^{-1}\mvar{\tilde{P}}\mvar U$, we write the problem
equivalently as

\[
\min_{\vx\in\redgevec}\normInf{\vvar{\alpha_{0}}+\mvar P\vx}
\]
where the output maximum flow is 
\[
\flow(\vx)=\capacityMatrix(\vvar{\alpha}_{0}+\mvar P\varVec)/\normInf{\capacityMatrix(\vvar{\alpha}_{0}+\mvar P\varVec)}.
\]

\subsection{An Approximate Maximum Flow Algorithm}

Since the gradient descent method requires the objective function to be differentiable, we introduce a smooth version of $\norm{\cdot}_{\infty}$ which we
call $\smax_{t}$. In next section, we prove that there is a convex differentiable function $\smax_t$ such that $\gradient \smax_t$ is Lipschitz continuous with Lipschitz constant $\frac{1}{t}$ and such that
\[
\forall \vx \in \Redgevec
\enspace : \enspace
\normInf{\vx}-t\ln(2m)\leq\smax_{t}(\vx)\leq\normInf{\vx}
\enspace.
\]
Now we consider the following regularized optimization problem 
\[
\min_{\vx\in\redgevec}g_{t}(\vx)\enspace\text{ where }\enspace g_{t}(\vx)=\smax_{t}(\vec{\alpha_{0}}+\mvar P\vx).
\]
For the rest of this section, we consider solving this optimization
problem using gradient descent under $\normInf{\cdot}$.

First, we bound the Lipschitz constant of the gradient of $g_{t}$.

\begin{lemma} \label{lem:lip_constant_of_g}The gradient of $g_{t}$
is Lipschitz continuous with Lipschitz constant $\funLip=\frac{\norm{\mvar P}_{\infty}^{2}}{t}$
\end{lemma}

\begin{proof} By \lemmaref{lem:lipeq} and the Lipschitz continuity of $\gradient \smax_t$, we have 
\[
\smax_{t}(\vy)\leq\smax_{t}(\vx)+\innerProduct{\nabla\smax_{t}(\vx)}{\vy-\vx}+\frac{1}{2t}\normInf{\vy-\vx}.
\]
Seting $\vx\leftarrow\vec{\alpha_{0}}+\mvar P\vx$ and $\vy\leftarrow\vec{\alpha_{0}}+\mvar P\vy$,
we have 
\begin{eqnarray*}
g_{t}(\vy) & \leq & g_{t}(\vy)+\innerProduct{\nabla\smax_{t}(\vec{\alpha_{0}}+\mvar P\vx)}{\mvar P\vy-\mvar P\vx}+\frac{1}{2t}\normInf{\mvar P\vy-\mvar P\vx}^{2}\\
 & \leq & g_{t}(\vy)+\innerProduct{\mvar P^{T}\nabla\smax_{t}(\vec{\alpha_{0}}+\mvar P\vx)}{\vy-\vx}+\frac{1}{2t}\norm{\mvar P}_{\infty}^{2}\normInf{\vy-\vx}^{2}\\
 & = & g_{t}(\vy)+\innerProduct{\nabla g_{t}(\vx)}{\vy-\vx}+\frac{1}{2t}\norm{\mvar P}_{\infty}^{2}\normInf{\vy-\vx}^{2}.
\end{eqnarray*}
Hence, the result follows from Lemma \ref{lem:lipeq}.\end{proof}

Now, we apply gradient descent to find an approximate max flow as follows. 

\begin{center}
\begin{tabular}{|l|}
\hline 
\textbf{MaxFlow}\tabularnewline
\hline 
\hline 
Input: any initial feasible flow $\vec{f_{0}}$ and $\text{OPT}=\min_{\vx}\normInf{\capacityMatrix^{-1}\flow_{0}+\mvar P\vx}$.\tabularnewline
\hline 
1. Let $\vec{\alpha_{0}}=\left(\mvar I-\mvar P\right)\capacityMatrix^{-1}\flow_{0}$
and $\vec{x_{0}}=0$.\tabularnewline
\hline 
2. Let $t=\varepsilon\text{OPT}/2\ln(2m)$ and $k=300\norm{\mvar P}_{\infty}^{4}\ln(2m)/\varepsilon^{2}$. \tabularnewline
\hline 
3. Let $g_{t}=\smax_{t}(\vec{\alpha_{0}}+\mvar P\vx)$.\tabularnewline
\hline 
4. For $i=1,\cdots,k$\tabularnewline
\hline 
5. $\quad\vx_{i+1}=\vx_{i}-\frac{t}{\norm{\mvar P}_{\infty}^{2}}\dualVec{(\gradient g_{t}(\vx_{i}))}.$
(See Lemma \ref{lem:sharp_formula_Linf_})\tabularnewline
\hline 
6. Output $\capacityMatrix\left(\vec{\alpha_{0}}+\mvar P\vx_{k}\right)/\normInf{\vec{\alpha_{0}}+\mvar P\vx_{k}}$.\tabularnewline
\hline 
\end{tabular}
\par\end{center}

We remark that the initial flow can be obtained by BFS and the OPT value can be approximted using binary search. In Section \ref{sec:nonlinear_projection}, we will give an algorithm with better dependence on $\norm{\mvar P}$.

\begin{theorem} \label{thm:MaxFlowAlgorithm}Let $\mProj$
be a cycle projection matrix, let $\mvar \mProjScaled = \mvar \capacityMatrix^{-1} \mProj \capacityMatrix$, and let $\varepsilon<1$. \textbf{MaxFlow} outputs an $(1 - \varepsilon)$-approximate
maximum flow in time 
\[
O\left(\frac{\norm{\mvar P}_{\infty}^{4}\ln(m)\left(\runtime(\mvar P)+m\right)}{\varepsilon^{2}}\right).
\]
\end{theorem} \begin{proof} First, we bound $\normInf{\vec{\alpha_{0}}}$.
Let $\vec{x}^{*}$ be a minimizer of $\min_{\vx}\normInf{\capacityMatrix^{-1}\flow_{0}+\mvar P\vx}$
such that $\mvar P\vec{x}^{*}=\vec{x}^{*}$. Then, we have 
\begin{eqnarray*}
\normInf{\vec{\alpha_{0}}} & = & \normInf{\capacityMatrix^{-1}\flow_{0}-\mvar P\capacityMatrix^{-1}\flow_{0}}\\
 & \leq & \normInf{\capacityMatrix^{-1}\flow_{0}+\vec{x}^{*}}+\normInf{\vec{x}^{*}+\mvar P\capacityMatrix^{-1}\flow_{0}}\\
 & = & \normInf{\capacityMatrix^{-1}\flow_{0}+\vec{x}^{*}}+\normInf{\mvar P\vec{x}^{*}+\mvar P\capacityMatrix^{-1}\flow_{0}}\\
 & \leq & \left(1+\normInf{\mvar P}\right)\normInf{\capacityMatrix^{-1}\flow_{0}+\vec{x}^{*}}\\
 & = & \left(1+\normInf{\mvar P}\right)\text{OPT}.
\end{eqnarray*}

Second, we bound $R$ in Theorem \ref{thm:gradient_descent}. Note
that 
\[
g_{t}(\vx_{0})=\smax_{t}(\vec{\alpha_{0}})\leq\normInf{\vec{\alpha_{0}}}\leq\left(1+\normInf{\mvar P}\right)\text{OPT}.
\]
Hence, the condition $g_{t}(\vx)\leq g_{t}(\vx_{0})$ implies that
\[
\normInf{\vec{\alpha_{0}}+\mvar P\vx}\leq\left(1+\normInf{\mvar P}\right)\text{OPT}+t\ln(2m).
\]
For any $\vy\in X^{*}$ let $\vc=\vx-\mvar P\vx+\vy$ and note
that $\mvar P\vc=\mvar P\vy$ and therefore $\vc\in X^{*}$. Using
these facts, we can bound $R$ as follows 
\begin{align*}
R & =\max_{\vx\in\redgevec~:~g_{t}(\vx)\leq g_{t}(\vx_{0})}\left\{ \min_{\vx^{*}\in X^{*}}\normInf{\vx-\vx^{*}}\right\} \\
 & \leq\max_{\vx\in\redgevec~:~g_{t}(\vx)\leq g_{t}(\vx_{0})}\normInf{\vx-\vc}\\
 & \leq\max_{\vx\in\redgevec~:~g_{t}(\vx)\leq g_{t}(\vx_{0})}\normInf{\mvar P\vx-\mvar P\vy}\\
 & \leq\max_{\vx\in\redgevec~:~g_{t}(\vx)\leq g_{t}(\vx_{0})}\normInf{\mvar P\vx}+\normInf{\mvar P\vy}\\
 & \leq2\normInf{\vec{\alpha_{0}}}+\normInf{\vec{\alpha_{0}}+\mvar P\vx}+\normInf{\vec{\alpha_{0}}+\mvar P\vy}\\
 & \leq2\normInf{\vec{\alpha_{0}}}+2\normInf{\vec{\alpha_{0}}+\mvar P\vx}\\
 & \leq4\left(1+\normInf{\mvar P}\right)\text{OPT}+2t\ln(2m).
\end{align*}
From Lemma \ref{lem:lip_constant_of_g}, we know that the Lipschitz
constant of $\gradient g_t$ is $\norm{\mvar P}_{\infty}^{2}/t$. Hence, Theorem
\ref{thm:gradient_descent} shows that
\begin{eqnarray*}
g_{t}(\vx_{k}) & \leq & \min_{\vx}g_{t}(\vx)+\frac{2\cdot L\cdot R^{2}}{k+4}\\
 & \leq & \text{OPT}+\frac{2\cdot L\cdot R^{2}}{k+4}.
\end{eqnarray*}
So, we have
\begin{eqnarray*}
\normInf{\vec{\alpha_{0}}+\mvar P\vx_{k}} & \leq & g_{t}(\vx_{k})+t\ln(2m)\\
 & \leq & \text{OPT}+t\ln(2m)+\frac{2\norm{\mvar P}_{\infty}^{2}}{t(k+4)}\left(4\left(1+\normInf{\mvar P}\right)\text{OPT}+2t\ln(2m)\right)^{2}.
\end{eqnarray*}
Using $t=\varepsilon\text{OPT}/2\ln(2m)$ and $k=300\norm{\mvar P}_{\infty}^{4}\ln(2m)/\varepsilon^{2}$,
we have 
\[
\normInf{\vec{\alpha_{0}}+\mvar P\vx_{k}}\leq(1+\varepsilon)\text{OPT}.
\]
Therefore, $\vec{\alpha_{0}}+\mvar P\vx_{k}$ is an $(1-\varepsilon)$
approximate maximum flow.

Now, we estimate the running time. In each step $5$, we are required
to compute $\dualVec{(\gradient g(\vx_{k}))}$. The gradient 
\[
\nabla g(\vx)=\mvar P^{T}\nabla\smax_{t}(\vec{\alpha_{0}}+\mvar P\vx)
\]
can be computed in $O(\runtime(\mvar P)+m)$ using the formula of
the gradient of $\smax_{t}$, applications of $\mvar P$ and $\mvar P^{T}$.
Lemma \ref{lem:sharp_formula_Linf_} shows that the $\#$ operator
can be computed in $O(m)$. \end{proof}

\begin{lemma} \label{lem:sharp_formula_Linf_}In $\normInf{\cdot}$,
the $\#$ operator is given by the explicit formula
\[
\left(\dualVec{\vx}\right)_{e}=\mathrm{sign}(x_{e})\norm{\vx}_{1}\quad\text{for }e\in E.
\]
\end{lemma} \begin{proof} Recall that 
\begin{eqnarray*}
\dualVec{\vx} & = & \argmax_{\vs\in\R}\innerProduct{\vx}{\vs}-\frac{1}{2}\normInf{\vec{s}}^{2}.
\end{eqnarray*}

It is easy to see that for all $e\in E$, $||\dualVec{\vx}||_{\infty}=\left|\left(\dualVec{\vx}\right)_{e}\right|$
. In particular, we have 
\[
\left(\dualVec{\vx}\right)_{e}=\mathrm{sign}(x_{e})||\dualVec{\vx}||_{\infty}.
\]
Fact \ref{fact:dualnorm_norm} shows that $||\dualVec{\vx}||_{\infty}=\norm{\vx}_{1}$
and the result follows.\end{proof}

\subsection{Properties of soft max}

In this section, we define $\smax_{t}$ and discuss its properties.
Formally, the regularized convex function can be found by smoothing
technique using convex conjugate \cite{nesterov2005smooth} \cite[Sec 5.4]{bertsekas1999nonlinear}.
For simplicity and completeness, we define it explicitly and prove
its properties directly. Formally, we define 
\[
\forall\vx\in\redgevec,\forall t\in\rPos\enspace:\enspace\smax_{t}(\vx)\defeq t\ln\left(\frac{\sum_{e\in E}\exp\left(\frac{x_{e}}{t}\right)+\exp\left(-\frac{x_{e}}{t}\right)}{2m}\right).
\]
For notational simplicity, for all $\varVec$ where this vector is
clear from context, we define $\vvar c$ and $\vvar s$ as follows
\[
\forall e\in E\enspace:\enspace\vvar c_{e}\defeq\exp\left(\frac{x_{e}}{t}\right)+\exp\left(-\frac{x_{e}}{t}\right)\enspace\text{ and }\enspace\vvar s_{e}\defeq\exp\left(\frac{x_{e}}{t}\right)-\exp\left(-\frac{x_{e}}{t}\right),
\]
where the letters are chosen due to the very close resemblance to
hyperbolic sine and hyperbolic cosine.

\begin{lemma}\label{lem:smax_grad} 
\[
\forall\varVec\in\R^{n}\enspace:\enspace\gradient\smax_{t}(\varVec)=\frac{1}{\onesVec^{T}\vvar c}\vvar s
\]

\[
\forall\varVec\in\R^{n}\enspace:\enspace\gradient^{2}\smax_{t}(\varVec)=\frac{1}{t\left(\onesVec^{T}\vvar c\right)}\left[\diag(\vvar c)-\frac{\vvar s\vvar s^{T}}{\onesVec^{T}\vvar c}\right]
\]
\end{lemma}

\begin{proof} For all $i\in E$ and $\vx\in\redgevec$, we have 
\begin{align*}
\frac{\partial}{\partial x_{i}}\smax_{t}(\varVec) & =\frac{\partial}{\partial x_{i}}\left(t\ln\left(\frac{\sum_{e\in E}\exp\left(\frac{x_{e}}{t}\right)+\exp\left(-\frac{x_{e}}{t}\right)}{2m}\right)\right)\\
 & =\frac{\exp\left(\frac{x_{i}}{t}\right)-\exp\left(-\frac{x_{i}}{t}\right)}{\sum_{e\in E}\exp\left(\frac{x_{e}}{t}\right)+\exp\left(-\frac{x_{e}}{t}\right)}.
\end{align*}
For all $i,j\in E$ and $\varVec\in\redgevec$, we have 
\begin{align*}
\frac{\partial^{2}}{\partial x_{i}\partial x_{j}}\smax_{t}(\varVec) & =\frac{\partial^{2}}{\partial x_{i}\partial x_{j}}\left(t\ln\left(\frac{\sum_{e\in E}\exp\left(\frac{x_{e}}{t}\right)+\exp\left(-\frac{x_{e}}{t}\right)}{2m}\right)\right)\\
 & =\frac{\partial}{\partial j}\left[\frac{\exp\left(\frac{x_{i}}{t}\right)-\exp\left(-\frac{x_{i}}{t}\right)}{\sum_{e\in E}\exp\left(\frac{x_{e}}{t}\right)+\exp\left(-\frac{x_{e}}{t}\right)}\right]\\
 & =\frac{1}{t}\frac{\left(\onesVec^{T}\vvar c\right)\indicVec{i=j}\left(\vvar c_{i}\right)-\vvar s_{i}\vvar s_{j}}{\left(\onesVec^{T}\vvar c\right)^{2}}.
\end{align*}
\end{proof}

\begin{lemma} \label{smax_properties}The function $\smax_{t}$ is
a convex continuously differentiable function and it has Lipschitz
continuous gradient with Lipschitz constant $1/t$ and 
\[
\normInf{\vx}-t\ln(2m)\leq\smax_{t}(\vx)\leq\normInf{\vx}
\]
for $\vx\in\redgevec$.\end{lemma}

\begin{proof} By the formulation of the Hessian, for all $\vx,\vy\in\redgevec$,
we have 
\[
\vy^{T}\left(\gradient^{2}\smax_{t}(\vx)\right)\vy\leq\frac{\sum_{i}c_{i}\vy_{i}^{2}}{t(\onesVec^{T}\vc)}\leq\frac{\sum_{i}c_{i}(\max_{j}\vy_{j}^{2})}{t(\onesVec^{T}\vc)}\leq\frac{1}{t}\normInf y^{2}.
\]
On the other side, for all $\vx,\vy\in\redgevec$, we have by $s_{i}\leq |s_i| \leq c_{i}$
and Cauchy Schwarz shows that 
\[
\vy^{T}\vvar s\vvar s^{T}\vy
\leq(\onesVec^{T}\vvar |s|)(\vy^{T}\diag(\vvar |s|)\vy).
\leq(\onesVec^{T}\vvar c)(\vy^{T}\diag(\vvar c)\vy).
\]
and hence 
\[
0\leq\vy^{T}\left(\gradient^{2}\smax_{t}(\vx)\right)\vy.
\]
Thus, the first part follows from Lemma \ref{lem:lip_formula}. For
the later part, we have 
\[
\normInf{\vx}\geq t\ln\left(\frac{\sum_{e\in E}\exp\left(\frac{x_{e}}{t}\right)+\exp\left(-\frac{x_{e}}{t}\right)}{2m}\right)\geq t\ln\left(\frac{\exp\left(\frac{\normInf{\vx}}{t}\right)}{2m}\right)=\normInf{\vx}-\ln(2m).
\]

\end{proof}

\section{Oblivious Routing}

In the previous sections, we saw how a circulation projection matrix can be used to solve max flow. In the next few sections, we show how to efficiently construct a circulation projection matrix to obtain an almost linear time algorithm for solving max flow.

Our proof focuses on the notion of (linear) oblivious routings. Rather than constructing the circulation projection matrix directly, we show how the efficient construction of an oblivious routing algorithm with a good competitive ratio immediately allows us to produce a circulation projection matrix.

In the remainder of this section, we formally define oblivious routings and prove the relationship between oblivious routing and circulation projection matrices (\sectionref{sec:oblivious}), provide a high level overview of our recursive approach and state the main theorems we will prove in later sections (\sectionref{sec:obliv_route:embed}). Finally, we prove the main theorem about our almost-linear-time construction of circulation projection  with norm $2^{O(\sqrt{\log(n)\log\log(n)})}$ assuming the proofs in the later sections (\sectionref{sec:obliv_route:construction_proof}). 

\subsection{From Oblivious Routing to Circulation Projection}
\label{sec:oblivious}

Here we provide definitions and prove basic properties of \emph{oblivious routings}, that is, fixed mappings from demands to flows that meet the input demands. While non-linear algorithms could be considered, we restrict our attention to linear oblivious routing strategies and use the term oblivious routing to refer to the linear subclass for the remainder of the paper.\footnote{Note that the oblivous routing strategies considered in \cite{KelnerMaymunkov09} \cite{LawlerNarayana09} \cite{Racke:2008:OHD:1374376.1374415} are all linear oblivious routing strategies.}

\begin{definition}[Oblivious Routing]\label{def:obliv_routing}
An \emph{oblivious routing} on graph $G = (V, E)$ is a linear operator $\mRoute \in \Rev$ such that for all demands $\demands \in \rvertvec$, $\incMatrix^T \mvar{A} \demands = \demands$. We call $\mRoute \demands$ the \emph{routing} of $\demands$ by $\mRoute$. 
\end{definition}

Oblivious routings get their name due to the fact that, given an oblivious routing strategy $\mRoute$ and a set of demands $D = \{\demands_1, \ldots, \demands_k\},$ one can construct a multicommodity flow satisfying all the demands in $D$ by using $\mRoute$ to route each demand individually, obliviously to the existence of the other demands. We measure the \emph{competitive ratio}\footnote{Again note that here and in the rest of the paper we focus our analysis on competitive ratio with respect to norm $\normInf{\cdot}$. However, many of the results present are easily generalizable to other norms. These generalizations are outside the scope of this paper.} of such an oblivious routing strategy to be the ratio of the worst relative congestion of such a routing to the minimal-congestion routing of the demands.

\begin{definition}[Competitive Ratio]\label{def:competitive_ratio}
The \emph{competitive ratio of oblivious routing $\mvar{A} \in \R^{E \times V}$}, denoted $\rho(\mvar{A})$, is given by
\[
\rho(\mRoute)
\defeq 
\max_{\{\demands_i\} ~ : ~ \forall i ~ \demands_i \perp \onesVec}
\frac{\congest(\{\mRoute \demands_i\})}{\opt(\{\demands_i\})}
\]
\end{definition}

At times, it will be more convenient to analyze an oblivious routing as a linear algebraic object rather a combinatorial algorithm; towards this end, we note that the competitive ratio of a linear oblivious routing strategy can be gleaned from the operator norm of a related matrix (see also \cite{LawlerNarayana09} and \cite{KelnerMaymunkov09}). Below, we state and prove a generalization of this result to weighted graphs that will be vital to relating $\mRoute$ to $\mProj$.

\begin{lemma} \label{lem:equivalence_of_competetivity_and_norm}
For any oblivious routing $\mvar{A},$ we have 
$
\rho(\mvar{A}) = \normInf{\capacityMatrix^{-1} \mvar{A} \incMatrix^T \capacityMatrix}
$
\end{lemma}

\begin{proof}
For a set of demands $D,$ let $D_\infty$ be the set of demands that results by taking the routing of every demand in $D$ by $\opt(D)$ and splitting it up into demands on every edge corresponding to the flow sent by $\opt(D)$. Now, clearly $\opt(D) = \opt(D_\infty)$ since routing $D$ can be used to route $D_\infty$ and vice versa, and clearly $\congest(\mvar{A} D) \leq \congest(\mvar{A} D_\infty)$ by the linearity of $\mvar{A}$ (routing $D_\infty$ simply doesn't reward $\mvar{A}$ routing for cancellations). Therefore,
\begin{align*}
\rho_p(\mvar{A})
&=
\max_{D} \frac{\congest(\{\mvar{A} D\})}{\opt(D)}
= 
\max_{D_\infty} \frac{\congest(\mvar{A} D_{\infty})}{\opt(D_\infty)}
= \max_{\vx \in \redgevec}
\frac{\norm{\sum_{e \in E} \vx_e \abs{\capacityMatrix^{-1} \mvar{A} \demandsEdge{e}}}_\infty}{\norm{\capacityMatrix^{-1} \vx}_\infty}
\\
&= 
\max_{\vx \in \redgevec}
\frac{\normInf{ \abs{\capacityMatrix^{-1} \mvar{A} \incMatrix^T} \vx}}{\normInf{\capacityMatrix^{-1} \vx}}
= 
\max_{\vx \in \redgevec}
\frac{\normInf{ \abs{\capacityMatrix^{-1} \mvar{A} \incMatrix^T \capacityMatrix} \vx}}{\normInf{\vx}}.
\end{align*}
\end{proof}

To make this lemma easily applicable in a variety of settings, we make use of the following easy to prove lemma.

\begin{lemma}[Operator Norm Bounds]
\label{lem:operator_norm_bounds}
For all $\varMat \in \R^{n \times m},$ we have that
\[
\normInf{\varMat}
= \normInf{ \abs{\varMat} }
= \normInf{ \abs{\varMat} \onesVec }
= \max_{i \in {n}} \norm{|\varMat|^T \indicVec{i}}_1
\]
\end{lemma}

The previous two lemmas make the connection between oblivious routings and circulation projection matrices clear. Below, we prove it formally.

\begin{lemma}[Oblivious Routing to Circulation Projection]
\label{lem:obl_rout_to_circulation_projection}
For oblivious routing $\mRoute \in \R^{E \times V}$ the matrix $\mProj \defeq \iMatrix- \mRoute \incMatrix^T$ is a circulation projection matrix such that
$
\normInf{\capacityMatrix \mProj \capacityMatrix^{-1}} \leq 1 + \rho(\mRoute) 
$
.
\end{lemma}

\begin{proof}
First, we verify that $\im(\mProj)$ is contained in cycle space:
\[
\forall \vx \in \redgevec
\enspace : \enspace
\incMatrix^T \mProj \vx
= \incMatrix^T \vx - \mRoute \incMatrix^T \vx
= \vzero.
\]
Next, we check that $\mProj$ is the identity on cycle space
\[
\forall \vx \in \redgevec \text{ s.t. } \incMatrix^T \vx = \vzero
\enspace : \enspace
\mProj \vx
= \vx - \mRoute \incMatrix^T \vx
= \vx.
\]
Finally, we bound the $\ell_\infty$-norm of the scaled projection matrix:
\[
\normInf{\capacityMatrix \mProj \capacityMatrix^{-1}}
=
\normInf{\iMatrix - \capacityMatrix \mRoute \incMatrix^T \capacityMatrix^{-1}}
\leq 1 + \rho(\mRoute).
\]
\end{proof}

\subsection{A Recursive Approach by Embeddings}
\label{sec:obliv_route:embed}

We construct an oblivious routing for a graph recursively. Given a generic, possibly complicated, graph, we show how to reduce computing an oblivious routing on this graph to computing an oblivious routing on a simpler graph on the same vertex set. A crucial concept in these constructions will be the notion of an embedding, which will allow us to relate the competitive ratios of an oblivious routing algorithms over graphs on the same vertex sets but different edge sets. 

\begin{definition}[Embedding]\label{def:embedding}
Let $G = (V, E, \capacityVec)$ and $G' = (V, E', \capacityVec')$ denote two undirected capacitated graphs on the same vertex set with incidence matrices $\incMatrix \in \R^{E \times V}$ and $\incMatrix' \in \R^{E' \times V}$ respectively. An \emph{embedding} from $G$ to $G'$ is a matrix $\mEmbed \in \R^{E' \times E}$ such that ${\incMatrix'}^T \mEmbed = \incMatrix^T$.
\end{definition}

In other words, an embedding is a map from flows in one graph $G$ to flows in another graph $G'$ that preserves the demands met by the flow. We can think of an embedding as a way of routing any flow in graph $G$ into graph $G'$ that has the same vertex set, but different edges. We will be particularly interested in embeddings that increase the congestion of the flow by a small amount going from $G$ to $G'$. 

\begin{definition}[Embedding Congestion]\label{def:embedding_cong}
Let $\mEmbed \in \R^{E' \times E}$ be an embedding from $G = (V, E, \capacityVec)$ to $G' = (V, E', \capacityVec')$ and let $\capacityMatrix \in \R^{E \times E}$ and $\capacityMatrix' \in \R^{E' \times E'}$ denote the capacity matrices of $G$ and $G'$ respectively. The \emph{congestion} of embedding $\mEmbed$ is given by
\[
\congest(\mEmbed) \defeq
\max_{\vx \in \R^{E}} \frac{\normInf{{\capacityMatrix'}^{-1} \mEmbed \vx}}{\normInf{\capacityMatrix^{-1} \vx}}
= \normInf{{\capacityMatrix'}^{-1} |\mEmbed| \capacityMatrix \onesVec}
\enspace.
\]
We say \emph{$G$ embeds into $G'$} with congestion $\alpha$ if there exists an embedding $\mEmbed$ from $G$ to $G'$ such that $\congest(\mEmbed) \leq \alpha$
\end{definition}

Embeddings potentially allow us to reduce computing an oblivious routing in a complicated graph to computing an oblivious routing in a simpler graph. Specifically, if we can embed a complicated graph in a simpler graph and we can efficiently embed the simple graph in the original graph, both with low congestion, then we can just focus on constructing oblivious routings in the simpler graph. We prove this formally as follows.

\begin{lemma}[Embedding Lemma]
\label{lem:embedding_lemma}
Let $G = (V, E, \capacityVec)$ and $G' = (V, E', \capacityVec')$ denote two undirected capacitated graphs on the same vertex sets, let $\mEmbed \in \R^{E' \times E}$ denote an embedding from $G$ into $G'$, let $\mEmbed' \in \R^{E \times E'}$ denote an embeding from $G'$ into $G$, and let $\mRoute' \in \R^{E' \times V}$ denote an oblivious routing algorithm on $G'$. Then $\mRoute \defeq \mEmbed' \mRoute'$ is an oblivious routing algorithm on $G$ and
\[
\rho(\mRoute)
\leq
\congest(\mEmbed) \cdot \congest(\mEmbed') \cdot \rho(\mRoute')
\]
\end{lemma}

\begin{proof}
For all $\vx \in \Rvertvec$ we have by definition of embeddings and oblivious routings that
\[
\incMatrix^T \mRoute \vx
= \incMatrix^T \mEmbed' \mRoute' \vx
= \incMatrix^T \vx.
\]
To bound $\rho(A),$ we let $\capacityMatrix$ denote the capacity matrix of $G$ and $\capacityMatrix'$ denote the capacity matrix of $G'$. Using \lemmaref{lem:equivalence_of_competetivity_and_norm}, we get
\[
\rho(\mRoute)
= \normInf{\capacityMatrix^{-1} \mRoute \incMatrix^T \capacityMatrix}
= \normInf{\capacityMatrix^{-1} \mEmbed' \mRoute' \incMatrix^T \capacityMatrix}
\]
Using that $\mEmbed$ is an embedding and therefore ${\incMatrix'}^T \mEmbed = \incMatrix^T$, we get
\[
\rho(\mRoute)
= \normInf{\capacityMatrix^{-1} \mEmbed' \mRoute' {\incMatrix'}^T \mEmbed\capacityMatrix}
\leq 
 \normInf{\capacityMatrix^{-1} \mEmbed' \capacityMatrix'}
 \cdot \normInf{{\capacityMatrix'}^{-1} \mRoute' {\incMatrix'}^T \capacityMatrix'}
 \cdot \normInf{{\capacityMatrix'}^{-1} \mEmbed \capacityMatrix}
\]
By the definition of competitive ratio and congestion, we obtain the result.
\end{proof}

Note how in this lemma we only use the embedding from $G$ to $G'$ to certify the quality of flows in $G'$, we do not actually need to apply this embedding in the reduction. 

Using this concept, we construct oblivious routings via recursive application of two techniques. First, in \sectionref{sec:flow_sparsifiers} we show how to take an arbitrary graph $G = (V, E)$ and approximate it by a \emph{sparse graph} $G' = (V, E')$ (i.e. one in which $|E'| = \otilde(|V|)$) such that flows in $G$ can be routed in $G'$ with low congestion and that there is an $\otilde(1)$ embedding from $G'$ to $G$ that can be applied in $\otilde(|E|)$ time. We call such a construction a \emph{flow sparsifiers} and prove the following theorem.

\begin{theorem}[Edge Sparsification]
\label{thm:edge_reduction}
Let $G = (V, E, \capacityVec)$ be an undirected capacitated graph with capacity ratio $U \leq \poly(|V|)$. In $\otilde(|E|)$ time we can construct a graph $G'$ on the same vertex set with at most $\otilde(|V|)$ edges and capacity ratio at most $\capacityRatio \cdot \poly(|V|).$ Moreover,  given an oblivious routing $\mRoute'$ on $G',$ in $\otilde(|E|)$ time we can construct an oblivious routing $\mRoute$ on $G$ such that
\[
\timeOf{\mRoute} = \otilde(|E| + \timeOf{\mRoute'})
\enspace \text{ and } \enspace
\rho(\mRoute) = \otilde(\rho(\mRoute'))
\]
\end{theorem}

Next, in \sectionref{sec:less_vertices} we show how to embed a graph into a collection of graphs consisting of trees plus extra edges. Then, we will show how to embed these graphs into better structured graphs consisting of trees plus edges so that by simply removing degree 1 and degree 2 vertices we are left with graphs with fewer vertices. Formally, we prove the following.

\begin{theorem}[Vertex Elimination]
\label{thm:node_reduction}
Let $G = (V, E, \capacityVec)$ be an undirected capacitated graph with capacity ratio $U$. For all $t > 0$ in $\otilde(t \cdot |E|)$ time we can compute graphs $G_1, \ldots , G_t$ each with at most $\otilde(\frac{|E| \log(U)}{t})$ vertices, at most $|E|$ edges, and capacity ratio at most $|V| \cdot U.$ Moreover, given oblivious routings $\mRoute_i$ for each $G_i$, in $\otilde(t \cdot |E|)$ time we can compute an oblivious routing $\mRoute$ on $G$ such that
\[
\timeOf{\mRoute} = \otilde(t \cdot |E| + \sum_{i = 1}^{t} \timeOf{\mRoute_i})
\enspace \text{ and } \enspace
\rho(\mRoute) = \otilde(\max_{i} \rho(\mRoute_i))
\]
\end{theorem}

In the next section we show that the careful application of these two ideas along with a powerful primitive for routing on constant sized graphs suffices to produce an oblivious routing with the desired properties.

\subsection{Efficient Oblivious Routing Construction Proof}
\label{sec:obliv_route:construction_proof}

First, we provide the lemma that will serve as the base case of our recursion. In particular, we show that electric routing can be used to obtain a routing algorithm with constant competitive ratio for constant-size graphs.

\begin{lemma}[Base Case] \label{thm:base_case}
Let $G = (V, E, \capacityVec)$ be an undirected capacitated graph and let us assign weights to edges so that $\wMatrix = \capacityMatrix^{2}$. For $\lap \defeq \incMatrix^T \wMatrix \incMatrix$ we have that $\mRoute \defeq \wMatrix \incMatrix \lapPseudo$ is an oblivious routing on $G$ with
$
\rho(\mRoute) \leq \sqrt{|E|}
$ 
and $\timeOf{\lapPseudo} = \otilde(|E|)$.
\end{lemma}

\begin{proof}
To see that $\mRoute$ is an oblivious routing strategy we note that for any demands $\demands \in \Rvertvec$ we have $\incMatrix^T \mRoute = \lap \lapPseudo = \iMatrix$. To see bound $\rho(\mRoute)$ we note that by \lemmaref{lem:equivalence_of_competetivity_and_norm} and standard norm inequalities we have
\[
\rho(\mRoute) = \max_{\vx \in \redgevec}
\frac{
\normInf{\capacityMatrix^{-1} \wMatrix \incMatrix \lapPseudo \incMatrix^T \capacityMatrix \vx}
}{\normInf{\vx}}
\leq 
\max_{\vx \in \redgevec}
\frac{
\normTwo{\capacityMatrix \incMatrix \lapPseudo \incMatrix^T \capacityMatrix \vx}
}{\frac{1}{\sqrt{|E|}} \normTwo{\vx}}
= \sqrt{|E|} \cdot \normTwo{\capacityMatrix \incMatrix \lapPseudo \incMatrix^T \capacityMatrix}
\]
The result follows from the fact in \cite{Spielman:2008:GSE:1374376.1374456} that $\mvar{\Pi} \defeq \capacityMatrix \incMatrix \lapPseudo \incMatrix^T \capacityMatrix$ is an orthogonal projection, and therefore $\normTwo{\mvar{\Pi}} \leq 1$, and the fact in \cite{Spielman:Solver:DBLP:journals/corr/abs-cs-0607105,Koutis:2010:AOS:1917827.1918388,Koutis:2011:NLN:2082752.2082901,koszSolver} that $\timeOf{\lapPseudo} = \otilde(|E|)$.
\end{proof}

Assuming \theoremref{thm:edge_reduction} and \theoremref{thm:node_reduction}, which we prove in the next two sections, we prove that low-congestion oblivious routings can be constructed efficiently.

\begin{theorem}[Recursive construction] \label{thm:master_theorem} Given an undirected capacitated
graph $G=(V,E,\capacityVec)$ with capacity ratio $\capacityRatio$. Assume $U = \poly(|V|)$.
We can construct an oblivious routing algorithm $\mRoute$ on $G$ in time $$O(|E|2^{O(\sqrt{\log|V|\log\log|V|})})$$
such that
\[
\timeOf{\mRoute}=|E|2^{O(\sqrt{\log|V|\log\log|V|})}
\enspace\text{ and }\enspace
\rho(\mRoute)=2^{O(\sqrt{\log|V|\log\log|V|})}.
\]
\end{theorem}

\begin{proof} Let $c$ be the constant hidden in the exponent terms,
including $\tilde{O}(\cdot)$ and $\poly(\cdot)$ in \theoremref{thm:edge_reduction}
and \theoremref{thm:node_reduction}. Apply \theoremref{thm:edge_reduction}
to construct a sparse graph $G^{(1)}$, then apply Theorem \ref{thm:node_reduction}
with $t=\left\lceil 2^{\sqrt{\log|V|\log\log|V|}}\right\rceil $ to get $t$ graphs $G_{1}^{(1)},\cdots G_{t}^{(1)}$ such that each
graphs have at most $O\left(\frac{1}{t}|E|\log^{2c}|V|\log{U}\right)$
vertices and at most $U\cdot|V|^{2c}$ capacity ratio. 

Repeat this process on each $G_{i}^{(1)}$, it produces $t^{2}$ graphs
$G_{1}^{(2)},\cdots,G_{t^{2}}^{(2)}$. Keep doing this until all graphs
$G_{i}$ produced have $O(1)$ vertices. Let $k$ be the highest level
we go through in this process. Since at the $k$-th level the
number of vertices of each graph is at most 
$O\left(\frac{1}{t^{k}}|E|\log^{2kc}|V|\log^{2k}(U|V|^{2ck})\right)$
vertices, we have $k=O\left(\sqrt{\frac{\log|V|}{\log\log|V|}}\right)$.

On each graph $G_{i}$, we use Theorem \ref{thm:base_case} to get an oblivious routing algorithm
$\mvar A_{i}$ for each $G_{i}$ with
\[
\timeOf{\mRoute_{i}}=O(1)
\enspace\text{ and }\enspace
\rho(\mRoute_{i})=O(1).
\]
Then, the Theorem \ref{thm:node_reduction} and \ref{thm:edge_reduction}
shows that we have an oblivious routing algorithm $\mvar A$ for
$G$ with 
\[
\timeOf{\mRoute}=O(tk|E|\log^{ck}(|V|)\log^{2k}(U|V|^{2ck}))
\enspace\text{ and }\enspace
\rho(\mRoute)=O(\log^{2kc}|V|\log^{k}(U|V|^{2ck})).
\]
The result follows from $k=O\left(\sqrt{\frac{\log|V|}{\log\log|V|}}\right)$
and $t=\left\lceil 2^{\sqrt{\log|V|\log\log|V|}}\right\rceil $.
\end{proof}

Using \theoremref{thm:master_theorem}, \lemmaref{lem:obl_rout_to_circulation_projection} and \theoremref{thm:MaxFlowAlgorithm}, we have the following almost linear time max flow algorithm on undirected graph.

\begin{theorem} \label{thm:maxflow_algorithm} Given an undirected capacitated graph $G=(V,E,\capacityVec)$
with capacity ratio $U$. Assume $U=\poly(|V|)$. There is an algorithm
finds an $(1-\varepsilon)$ approximate maximum flow in
time 
\[
O\left(\frac{|E|2^{O\left(\sqrt{\log|V|\log\log |V|}\right)}}{\varepsilon^{2}}\right).
\]
\end{theorem}

\section{Flow Sparsifiers}
\label{sec:flow_sparsifiers}

In order to prove Theorem~\ref{thm:edge_reduction}, i.e. reduce the problem of efficiently computing a competitive oblivious routing on a dense graph to the same problem on a sparse graph, we introduce a new algorithmic tool called \emph{flow sparsifiers}. \footnote{Note that our flow sparsifiers aim to reduce the number of edges, and are different from the flow sparsifiers of Leighton and Moitra~\cite{LeightonMoitra}, which work in a different setting and reduce the number of vertices.} A flow sparsifier is an efficient cut-sparsification algorithm that also produces an efficiently-computable low-congestion embedding mapping the sparsified graph back to the original graph.

\begin{definition}[Flow Sparsifier] \label{def:flowsparsify}
An algorithm is a \emph{$(h, \epsilon, \alpha)$-flow sparsifier} if on input graph $G = (V, E, \mu)$ with capacity ratio $U$ it outputs a graph $G' = (V, E', \mu')$ with capacity ratio $U' \leq U \cdot \poly(|V|)$ and an embedding $\mEmbed: \R^{E'} \rightarrow \R^{E}$ of $G'$ into $G$  with the following properties:
\begin{itemize}
  \item \textbf{Sparsity:} $G'$ is $h$-sparse, i.e.
\[
|E'| \leq h
\]
\item \textbf{Cut Approximation:} $G'$ is an $\epsilon$-cut approximation of $G,$ i.e.
\[
\forall S \subseteq V
\enspace : \enspace
(1 - \epsilon) \mu(\cutset_G(S)) \leq \mu'(\cutset_{G'}(S)) \leq (1 + \epsilon) \mu(\cutset_G(S))
\]
\item \textbf{Flow Approximation:} $\mEmbed$ has congestion at most $\alpha,$ i.e.
\[
\congest(\mEmbed) \leq \alpha.
\]
\item \textbf{Efficiency:} The algorithm runs in $\tilde{O}(m)$ time and $\timeOf{\mEmbed}$ is also $\tilde{O}(m)$.
\end{itemize}
\end{definition}

Flow sparsifiers allow us to solve a multi-commodity flow problem on a possibly dense graph $G$ by converting $G$ into a sparse graph $G'$ and solving the flow problem on $G'$, while suffering a loss of a factor of at most $\alpha$ in the congestion when mapping the solution back to $G$ using $\mEmbed.$

\begin{theorem} \label{thm: flow-sparsify-conseq}
Consider a graph $G=(V,E,\mu)$ and let $G'=(V,E', \mu')$ be given by an $(h, \epsilon, \alpha)$-flow sparsifier of $G.$ 
Then, for any set of $k$ demands $D=\{\demands_1, \demands_2, \ldots, \demands_k\}$ between vertex pairs of $V,$
we have: 
\begin{equation} \label{eq: flow-sparsify-conseq1}
\opt_{G'}(D) \leq \frac{O(\log k)}{1-\epsilon} \cdot \opt_{G}(D). 
\end{equation}
Given the optimum flow $\{f^\star_i\}$ over $G'$, we have 
$$
\congest_G(\{\mEmbed f^\star_i\}) \leq \alpha \cdot \opt_{G'}(D) \leq \frac{O(\alpha \log k)}{1-\epsilon} \cdot \opt_G(D).
$$
\end{theorem}
\begin{proof}
By the flow-cut gap theorem of Aumann and Rabani~\cite{flowcutgap}, we have that, for any set of $k$ demands $D$ on $V$ we have:
$$
\opt_{G}(D) \geq O\left(\frac{1}{\log k}\right) \cdot \max_{S \subset V} \frac{D(\cutset(S))}{\mu(\cutset_{G}(S))}. 
$$
where $D(\cutset(S))$ denotes the total amount of demand separated by the cut between $S$ and $\bar{S}.$
As any cut $S \subseteq V$ in $G'$ has capacity $\mu'(\cutset_{G'}(S)) \geq (1-\epsilon) \mu(\cutset_G(S)),$ we have:
$$
\opt_{G'}(D) \leq \max_{S \subset V} \frac{D(\cutset(S))}{\mu'(\cutset_{G'}(S))} \leq  
\frac{1}{1 - \epsilon} \cdot \max_{S \subset V} \frac{D(\cutset(S))}{\mu(\cutset_{G}(S))} \leq \frac{O(\log k)}{1-\epsilon} \cdot \opt_{G}(D).
$$
The second part of the theorem follows as a consequence of the definition of the congestion of the embedding $\mEmbed.$
\end{proof}

Our flow sparsifiers should be compared with the cut-based decompositions of 
R\"{a}cke \cite{Racke:2008:OHD:1374376.1374415}. R\"{a}cke constructs a probability distribution over trees and gives explicit embeddings from $G$ to this distribution and backwards, achieving a congestion of $O(\log n).$ However, this distribution over tree can include up to $O(m\log n)$ trees and it is not clear how to use it to obtain an almost linear time algorithm. Flow sparsifiers answer this problem by embedding $G$ into a single graph $G'$, which is larger than a tree, but still sparse. Moreover, they provide an explicit efficient embedding of $G'$ into $G.$ Interestingly, the embedding from $G$ to $G'$ is not necessary for our notion of flow sparsifier, and is replaced by the cut-approximation guarantee. This requirement, together with the application of the flow-cut gap~\cite{flowcutgap}, lets us argue that the optimal congestion of a $k$-commodity flow problem can change at most by a factor of $O(\log k)$ between $G$ and $G'.$

\subsubsection{Main Theorem on Flow Sparsifiers and Proof of Theorem~\ref{thm:edge_reduction}}

The main goal of this section will be to prove the following theorem:
\begin{theorem}\label{thm:flow-sparsify}
For any constant $\epsilon \in (0,1),$ there is an $(\tilde{O}(n), \epsilon, \tilde{O}(1))$-flow sparsifier.
\end{theorem}

Assuming Theorem~\ref{thm:flow-sparsify}, we can now prove Theorem~\ref{thm:edge_reduction}, the main theorem necessary for edge reduction in our construction of low-congestion projections.
\begin{proof}[Proof of Theorem~\ref{thm:edge_reduction}]
We apply the flow sparsifier of Theorem~\ref{thm:flow-sparsify} to $G=(V,E,\capacityVec)$ and obtain output $G'=(V,E', \capacityVec')$ with embedding $\mvar{M}.$ By the definition of flow sparsifier, we know that the capacity ratio $U'$ of $G'$ is at most $U \cdot \poly(|V|),$ as required. Moreover, again by Theorem~\ref{thm:flow-sparsify}, $G'$ has at most $\tilde{O}(|V|)$ edges. 
Given an oblivious routing $\mvar{A'}$ on $G'$ consider the oblivious routing $\mvar{A} \defeq \mvar{M} \mvar{A'}.$ By the definition of flow sparsifier, we have that $\timeOf{\mEmbed} = \otilde(|E|).$ Hence $\timeOf{\mRoute} = \timeOf{\mEmbed} + \timeOf{\mRoute'} =  \otilde(|E|) + \timeOf{\mRoute'}. $\

To complete the proof, we bound the competivite ratio $\rho(\mRoute).$ Using the same argument as in Lemma~\ref{lem:equivalence_of_competetivity_and_norm}, we can write $\rho(\mRoute)$ as
$$
\rho(\mvar{A})
=
\max_{D} \frac{\congest_G(\{\mvar{A} D\})}{\opt_G(D)}
\leq 
\max_{D_\infty} \frac{\congest_G(\mvar{A} D_{\infty})}{\opt_G(D_\infty)},
$$
where $D_\infty$ is the set of demands that result by taking the routing of every demand in $D$ by $\opt(D)$ and splitting it up into demands on every edge corresponding to the flow sent by $\opt(D)$. Notice that $D_{\infty}$ has at most $|E|$ demands that are routed between pairs of vertices in $V.$ Then, because $G'$ is an $\epsilon$-cut approximation to $G,$ the flow-cut gap of Aumann and Rabani~\cite{flowcutgap} guarantees that
$$
\opt_{G}(D_\infty) \geq \frac{1}{O(\log n)} \opt_{G'}(D_\infty).
$$
As a result, we obtain:
\begin{align*}
\rho(\mRoute) & \leq O(\log n) \cdot \max_{D_\infty} \frac{\congest_G(\mvar{A}D_{\infty})}{\opt_{G'} (D_\infty)} 
= O(\log n) \cdot \max_{D_\infty} \frac{\congest_G(\mvar{M} \mvar{A'} D_{\infty})}{\opt_{G'} (D_\infty)}\\
& \leq O(\log n) \cdot \congest(\mvar{M}) \cdot \max_{D_\infty} \frac{\congest_{G'}(\mvar{A'} D_{\infty})}{\opt_{G'} (D_\infty)}
\leq \otilde(\rho(\mRoute')).
\end{align*}
\end{proof}

\subsubsection{Techniques}

We will construct flow sparsifiers by taking as a starting point the construction of spectral sparsifiers of Spielman and Teng~\cite{STspectralSparse}. Their construction achieves a sparsity of $\otilde\left(\frac{n}{\epsilon^{2}}\right)$ edges, while guaranteeing an $\epsilon$-spectral approximation. As the spectral approximation implies the cut approximation, the construction in \cite{STspectralSparse} suffices to meet the first two conditions in Definition~\ref{def:flowsparsify}.
Moreover, their algorithm also runs in time  $\tilde{O}(m),$ meeting the fourth condition. Hence, to complete the proof of Theorem~\ref{thm:flow-sparsify}, we will modify the construction of Spielman and Teng to endow their sparsifier $G'$ with an embedding $\mEmbed$ onto $G$ of low congestion that can be both computed and invoked efficiently. The main tool we use in constructing $\mEmbed$ is the notion of electrical-flow routing and the fact that electrical-flow routing schemes achieve a low competitive ratio on near-expanders and subsets thereof~\cite{KelnerMaymunkov09, LawlerNarayana09}.

To exploit this fact and construct a flow sparsifier, we follow Spielman and Teng~\cite{STspectralSparse} and partition the input graph into vertex sets, where each sets induces a near-expanders and most edges of the graph do not cross set boundaries. We then sparsify these induced subgraphs using standard sparsification techniques and iterate on the edges not in the subgraphs. As each iteration removes a constant fraction of the edges, by using standard sparsification techniques, we immediately obtain the sparsity and cut approximation properties. To obtain the embedding $\mEmbed$ with $\congest(\mEmbed) = \otilde(1),$ we prove a generalization of results in~\cite{KelnerMaymunkov09, LawlerNarayana09}  and show that the electrical-flow routing achieves a low competitive ratio on near-expanders and subsets thereof. 

In the next two subsections, we introduce the necessary concept about electrical-flow routing and prove that it achieves low competitive ratio over near-expanders (and subsets of near-expanders).

\subsection{Subgraph Routing}

Given an oblivious routing strategy $\mvar{A},$ we may be interested only in routing demands coming from a subset of edge $F \subseteq E$. In this setting, given a set of demands $D$ routable in $F,$ we let $\opt^F(D)$ denote the minimal congestion achieved by any routing restricted to only sending flow on edges in $F$ and we measure the $F$-competitive ratio of $\mvar{A}$ by
\[
\rho^F(\mvar{A})
\defeq \max_{\text{$D$ routable in $F$}}
\frac{\congest(\mvar{A}D)}{\opt^F(D)}
\]
Note that $\mvar{A}$ may use all the edges in $G$ but $\rho^F(\mvar{A})$ compares it only against routings that are restricted to use only edges in $F$. As before, we can upper bound the $F$-competitive ratio $\rho^F(\mvar{A})$ by operator norms.
\begin{lemma} \label{lem:subrouting}
Let $\indicVec{F} \in \redgevec$ denote the indicator vector for set $F$ (i.e. $\indicVec{F}(e) = 1$ if $e \in F$ and $\indicVec{F}(e) = 0$) and let $\iMatrix_F \defeq diag(\indicVec{F}).$ For any $F \subseteq E$ we have
\[
\rho^F(\mvar{A})
= \norm{\capacityMatrix^{-1} \mvar{A} \incMatrix^T \capacityMatrix \iMatrix_F}_{\infty} 
\]
\end{lemma}

\begin{proof}
We use the same reasoning as in the non-subgraph case. For a set of demands $D = \{\demands_i\},$ we consider $D^F_\infty$, the demands on the edges in $F$ used by $\opt^F(D).$ Then, it is the case that $\opt^F(D) = \opt^F(D_\infty)$ and we know that cost of obliviously routing $D_P$ is greater than the cost of obliviously routing $D$. Therefore,  we have:
\begin{align*}
\rho^F
&= 
\max_{\varVec \in \redgevec ~ : ~ \iMatrix_{E \setminus F} \varVec = 0}
\frac
{\norm{
~
\sum_{e \in E}
|\capacityMatrix^{-1} \mvar{A} \incMatrix^T \indicVec{e} \varVec_e|
~
}_\infty
}
{\norm{\capacityMatrix^{-1} \varVec}_\infty}
\\
&= 
\max_{\varVecB \in \redgevec ~ : ~ \iMatrix_{E \setminus F} \varVecB = 0}
\frac
{\norm{
~
\sum_{e \in E}
|\capacityMatrix^{-1} \mvar{A} \incMatrix^T \capacityMatrix \indicVec{e} \varVecB_e|
~
}_\infty
}
{\norm{\varVecB}_\infty}
\\
& =
\max_{\varVecB \in \redgevec}
\frac
{\norm{
~
\sum_{e \in E}
|\capacityMatrix^{-1} \mvar{A} \incMatrix^T \capacityMatrix \iMatrix_F \indicVec{e} \varVecB_e|
~
}_\infty
}
{\norm{\varVecB}_\infty}
\tag{Having $y_e \neq 0$ for $e \in E \setminus F$ decreases the ratio.}
= \norm{\abs{\capacityMatrix^{-1} \mvar{A} \incMatrix^T \capacityMatrix \iMatrix_F}}_{\infty} = \norm{\capacityMatrix^{-1} \mvar{A} \incMatrix^T \capacityMatrix \iMatrix_F}_{\infty}
\\
\end{align*}
\end{proof}

\subsection{Electrical-Flow Routings}

In this section, we define the notion of electrical-flow routing and prove the results necessary to construct flow sparsifiers.
Recall that $\rMatrix$ is the diagonal matrix of resistances and the Laplacian $\lap$ is defined as $\incMatrix^T \rMatrix^{-1} \incMatrix.$ For the rest of this section, we assume that resistances are set as $\rMatrix = \capacityMatrix^{-1}.$

\begin{definition}\label{def:obl_elec}
Consider a graph $G=(V,E, \mu)$ and set the edge resistances as 
$r_e = \frac{1}{\mu_e}$ 
for all $e \in E.$
The oblivious {\emph electrical-flow routing strategy} is the linear operator $\mElRout$ defined as
$$
\mElRout \defeq \rMatrix^{-1} \incMatrix \lapPseudo,
$$
\end{definition}
In words, the electrical-flow routing strategy is the routing scheme that, for each demand $\demands$ sends the electrical flow with boundary condition $\demands$ on the graph $G$ with resistances $\rMatrix = \capacityMatrix^{-1}.$

For the electrical-flow routing strategy $\mElRout$, the upper bound on the competitive ratio $\rho(\mElRout)$ in Lemma~\ref{lem:equivalence_of_competetivity_and_norm} can be rephrased in terms of the voltages induced on $G$ by electrically routing an edge $e \in E.$ This interpretation appears in~\cite{KelnerMaymunkov09, LawlerNarayana09}.

\begin{lemma}\label{lem:elec_routing}
Let $\mElRout$ be the electrical-flow routing strategy. 
For an edge $e \in E,$ we let the voltage vector $\volt_e \in \rvertvec$ be given by $\volt_e \defeq\lapPseudo \demandsEdge{e}$. We then have
$$
\rho(\mElRout) = \max_{e \in E} \sum_{(a, b) \in E} \frac{\abs{v_e(a) - v_e(b)}}{r_{ab}}.
$$
\end{lemma}
\begin{proof}
We have: 
\[
\rho(\mElRout)
=
\norm{\mvar{\incMatrix \lapPseudo \incMatrix^T \rMatrix^{-1}}}_{\infty}
= \max_{e \in E} \norm{\mvar{\rMatrix^{-1} \incMatrix \lapPseudo \incMatrix^T} \indicVec{e}}_1 
= 
\max_{e \in E}
\sum_{(a, b) \in E} \frac{ \abs{v_e(a) - v_e(b)}}{r_{ab}}.
\]
\end{proof}

The same reasoning can be extended to the subgraph-routing case to obtain the following lemma.
\begin{lemma}\label{lem:elecsubrouting}
For $F \subseteq E$ and $\rMatrix = \capacityMatrix^{-1}$ we have
$$
\rho^F(\mElRout) = \max_{e \in E}
\sum_{(a, b) \in F} \frac{\abs{v_e(a) - v_e(b)}}{r_{ab}}.
$$
\end{lemma}

\begin{proof}
As before, we have:
\begin{align*}
\rho^F(\mElRout)
& = \norm{\incMatrix \lapPseudo \incMatrix^T \rMatrix^{-1} \iMatrix_F}_\infty \tag{By Lemma~\ref{lem:subrouting}} \\
& = \max_{e \in E} \norm{\iMatrix_F \rMatrix^{-1} \incMatrix \lapPseudo \incMatrix^T \indicVec{e}}_1 
= 
\max_{e \in E}
\sum_{(a, b) \in F} \frac{\abs{v_e(a) - v_e(b)}}{r_{ab}}
\end{align*}
\end{proof}

\subsubsection{Bounding the Congestion}

In this section, we prove that we can bound the $F$-competitive ratio of the oblivious electrical-routing strategy as long as the edges $F$ that the optimum flow is allowed to route over are contained within an induced expander $G(U) =(U, E(U))$ for some $U \subseteq V$ . Towards this we provide and prove the following lemma.
This is a generalization of a similar lemma proved in~\cite{KelnerMaymunkov09}.

\begin{lemma}\label{lem:elec_cong}
For weighted graph $G = (V, E, w)$ with integer weights and vertex subset $U \subseteq V$ the following holds: 
\[
\rho^F(\mElRout)
\leq \frac{8 \log(\vol(G(U)))}{\conductance(G(U))^2}
\]
\end{lemma}

\begin{proof}
By Lemma~\ref{lem:elecsubrouting}, for every edge $e \in E,$
$$
\rho^F(\mElRout) \leq \norm{\iMatrix_{E(U)} \rMatrix^{-1} \incMatrix \lapPseudo \demandsEdge{e}}_1 
$$
Fix any edge $e \in E$ and let $v \defeq \lapPseudo \demandsEdge{e}.$ Recall that with this definition
\begin{equation} \label{eq:electricalcongbound1}
\norm{\iMatrix_{E(U)} \rMatrix^{-1} \incMatrix \lapPseudo \demandsEdge{e}}_1
=
\sum_{(a, b) \in E(U)} \frac{|v(a) - v(b)|}{r_{ab}} = \sum_{(a, b) \in E(U)} w_{ab} \cdot {|v(a) - v(b)|}
\end{equation} 
We define the following vertex subsets:
\[
\forall x \in \R
\enspace : \enspace
S^{\leq}_x \defeq \{a \in U ~ | ~ v(a) \leq x \}
\enspace \text{ and } \enspace
S^{\geq}_x \defeq \{a \in U ~ | ~ v(a) \geq x \}
\]
Since adding a multiple of the all-ones vector to $v$  does not change the quantity of interest in Equation~\ref{eq:electricalcongbound1}, we can assume without loss of generality that
\[
\vol_{G(U)}(S^{\geq}_0) \geq \frac{1}{2} \left(\vol(G(U))\right)
\enspace \text{ and } \enspace
\vol_{G(U)}(S^{\leq}_0) \geq \frac{1}{2} \left(\vol(G(U))\right).
\]
For any vertex subset $S \subseteq U,$ we denote the flow out of $S$ and the weight out of $S$ by
\[
f(S) \defeq \sum_{e=(a,b) \in E(U) \bigcap \cutset(S)} w_e |v(a) - v(b)|,
\enspace \text{ and } \enspace
w(S) \defeq \sum_{e \in E(U) \bigcap \cutset(S)} w_e.
\]
At this point, we define a collections of subsets $\{C_i \in S^{\geq}_0\}$. For an increasing sequence of real numbers $\{c_i\},$ we let
$
C_i \defeq S_{c_i}^{\geq}
$ and  we define the sequence $\{c_i\}$ inductively as follows:
\[
c_0 \defeq 0
\enspace \text{ , } \enspace
c_i \defeq c_{i - 1} + \Delta_{i - 1}
\enspace \text{ , and } \enspace
\Delta_i \defeq 2 \frac{f(C_i)}{w(C_i)}
\enspace.
\]
In words, the $c_{i+1}$ equals the sum of $c_i$ and an increase $\Delta_i$ which depends on how much the cut $\delta(C_i) \cap E(U)$ was congested by the electrical flow.

Now,  $l_i \defeq w(\cutset_{E(U)}(C_{i - 1}) - \cutset_{E(U)}(C_{i}))$, i.e. the weight of the edges in $E(U)$ cut by $C_{i - 1}$ but not by $C_i$. We get
\begin{align*}
\vol(C_{i + 1})
&\leq
\vol(C_i) - l_i \\
&\leq
\vol(C_i) - \frac{w(C_i)}{2}
\tag{By choice of $l_i$ and $\Delta_i$} \\
&\leq
\vol(C_i) - \frac{1}{2} \vol(C_i) \conductance(G(U))
\tag{Definition of conductance}
\end{align*}
Applying this inductively and using our assumption on $\vol(S_0^{\geq})$
we have that
\[
\vol(C_i)
\leq
\left(1 - \frac{1}{2} \conductance(G(U))\right)^i \vol(C_0)
\leq
\frac{1}{2} \left(1 - \frac{1}{2} \conductance(G(U))\right)^i \vol(G(U))
\]
Since $\phi(G(U)) \in (0, 1),$  for $i + 1 = \frac{2 \log \left(\vol(G(U))\right)}{\conductance(G(U))}$ we have that $\vol(S_i) \leq \frac{1}{2}$. Since $\vol(S_i)$ decreases monotonically with $i,$ if we
let $r$ be the smallest value such that $C_{r + 1} = \emptyset,$ we must have
\[
r \leq \frac{2 \cdot \log(\vol(G(U)))}{\conductance(G(U))}
\]
Since $v$ corresponds to a unit flow, we know that $f(C_i) \leq 1$ for all $i.$ Moreover, by the definition of conductance we know that $w(C_i) \geq \conductance(G(U)) \cdot \vol(C_i)$. Therefore,
\[
\Delta_i \leq \frac{2}{\conductance(G(U)) \cdot \vol(C_i)}.
\]
We can now bound the contribution of $C_0^{\geq}$ to the volume of the linear embedding $v.$ In the following, for a vertex $a \in V,$ we let $d(a) \defeq \sum_{e=\{a,b\} \in E(U)} w_e$ be the degree of $a$ in $E(U)$.  
\begin{align*}
\sum_{a \in C^{\geq}_0}
d(a) v(a)
&= 
\sum_{i = 0}^{r} 
\left[
\sum_{a \in C_{i} - C_{i + 1}}
d(a) v(a) 
\right]
\\
&\leq
\sum_{i = 0}^{r} 
\left[
\sum_{a \in C_{i} - C_{i + 1}}
d(a) 
\left(
\sum_{j = 0}^{i}
\Delta_{j}
\right)
\right]
\tag{By definition of $C_i$}
\\
&\leq
\sum_{i = 0}^{r} 
\left[
\left(
\vol(C_{i}) - \vol(C_{i + 1})
\right)
\cdot
\left(
\sum_{j = 0}^{i}
\Delta_{j}
\right)
\right]
\\
&= \sum_{i = 0}^{r} \vol(C_i) \Delta_i
\tag{Rearrangement and fact that $\vol(C_{r + 1}) = 0$}
\leq \frac{2r}{\conductance(G(U))}
\end{align*}
By repeating the same argument on $S_0^{\leq},$ we get that 
$\sum_{a \in S_0^{\leq}} d(a) v(a) \leq \frac{2r}{\conductance(G(U))}$. 
Putting this all together yields
\[
\norm{\iMatrix_{E(U)} \rMatrix^{-1} \incMatrix \lapPseudo \demandsEdge{e}}
=
\sum_{(a,b) \in G(U)} w_{ab} \cdot |v(a) - v(b)|
\leq \sum_{a \in G(U)} d(a) v(a)
\leq \frac{4r}{\conductance(G(U))}
\]
\end{proof}

From this lemma and Lemma~\ref{lem:elecsubrouting}, the following is immediate:
\begin{lemma}\label{lem:elecsubroutinggoal}
Let $F \subseteq E$ be contained within some vertex induced subgraph $G(U),$ then 
for $\rMatrix = \capacityMatrix^{-1}$ we have
\[
\rho^F(\rMatrix^{-1} \incMatrix \lapPseudo)
\leq \rho^{E(U)}(\rMatrix^{-1} \incMatrix \lapPseudo)
\leq \frac{8 \log(\vol(G(U)))}{\conductance(G(U))^2}
\]
\end{lemma}

\subsection{Construction and Analysis of Flow Sparsifiers}

In the remainder of this section we show how to produce an efficient $O(\log^{c})$-flow sparsifier for some fixed constant $c,$ proving Theorem~\ref{thm:flow-sparsify}. In this version of the paper, we make no attempt to optimize the value of $c.$
For the rest of this section, we again assume that we choose the resistance of an edge to be the the inverse of its capacity, i.e. $\capacityMatrix = \wMatrix = \rMatrix^{-1}.$

As discussed before, our approach follows closely that of Spielman and Teng~\cite{STspectralSparse} to the construction of spectral sparsifiers. The first step of this line of attack is to reduce the problem to the unweighted case.

\begin{lemma}\label{lem:weightedtounweighted}
Given an $(h, \epsilon,\alpha)$-flow-sparsifier algorithm for unweighted graphs, it is possible to construct an $(h \cdot \log U, \epsilon, \alpha)$-flow-sparsifier algorithm for weighted graphs $G=(V,E,\mu)$ with capacity ratio $U$ obeying
$$
U = \frac{\max_{e \in E} \mu_e}{\min_{e \in E} \mu_e} = \poly(|V|).
$$
\end{lemma}
\begin{proof}
We write each edge in binary so that $G = \sum_{i=0}^{\log U} 2^i G_i$ for some unweighted graphs $\{G_i = (V, E_i\}_{i \in [\log U]}\},$ where $|E_i| \leq m$ for all $i$. We now apply the unweighted flow-sparsifier to each $G_i$ in turn to obtain graphs $\{G'_i\}.$ We let $G' \defeq \sum_{i=0}^{\log U} 2^i G'_i$ be the weighted flow-sparsified graph. 
By the assumption on the unweighted flow-sparsifier, each $G'_i$ is $h$-sparse,
so that $G'$ must have at most $h \cdot \log U$ edges.   
Similarly, $G'$ is an $\epsilon$-cut approximation of $G$, as each $G'_i$ is an $\epsilon$-cut approximation of the corresponding $G_i.$
Letting $\mvar{M}_i$ be the embedding of $G'_i$ into $G_i,$ we can consider the embedding $\mvar{M} = \sum_{i=0}^{\log U} 2^i \mvar{M}_i$ of $G'$ into $G.$
As each $\mvar{M_i}$ has congestion bounded by $\alpha,$ it must be the case that $\mvar{M}$ also has congestion bounded by $\alpha.$
The time to run the weighted flow sparsifier and to invoke $\mvar{M}$ is now $\tilde{O}(m) \cdot \log U = \tilde{O}(m)$ by our assumption on $U.$
\end{proof}

The next step is to construct a routine which flow-sparsifies a constant fraction of the edges of $E.$ This routine will then be applied iteratively to produce the final flow-sparsifier.
\begin{lemma}\label{lem:mainlemflowsparsify}
On input an unweighted graph $G=(V,E),$ there is an algorithm that runs in $\tilde{O}(m)$ and computes a partition of $E$ into $(F, \bar{F}),$ an edge set $F' \subseteq F$ with weight vector $\vec{w_{F'}} \in \R^E, \operatorname{support}(w_{F'}) = F',$ and an embedding $\mvar{H}:\R^{F'} \rightarrow \R^E$ with the following properties:
\begin{enumerate}
\item $F$ contains most of the volume of $G,$ i.e.
$$
|F| \geq \frac{|E|}{2};
$$
\item $F'$ contains only $\tilde{O}(n)$ edges, i.e. $|F'| \leq \tilde{O}(n).$
\item The weights $w_{F'}$ are bounded
$$
\forall e \in F' \enspace ,  \enspace \frac{1}{\poly(n)} \leq w_{F'}(e) \leq n.
$$
\item The graph $H'=(V,F', w_{F'})$ is an $\epsilon$-cut approximation to $H=(V,F),$ i.e. for all $S \subseteq V:$
$$
(1-\epsilon) |\cutset_H(S)|\leq w_{F'}(\cutset_{H'}(S)) \leq (1+\epsilon) |\cutset_H(S)|.
$$
\item The embedding $\mvar{H}$ from $H=(V,F', w_{F'})$ to $G$ has bounded congestion
$$
\congest(\mvar{H}) = \otilde(1).
$$
and can be applied in time $\tilde{O}(m).$
\end{enumerate}
\end{lemma}
Given Lemma~\ref{lem:weightedtounweighted} and Lemma~\ref{lem:mainlemflowsparsify}, it is straightforward to complete the proof of Theorem~\ref{thm:flow-sparsify}.
\begin{proof}
Using Lemma~\ref{lem:weightedtounweighted}, we reduce the objective of Theorem~\ref{thm:flow-sparsify} to running a $(\otilde(n), \epsilon, \otilde(1))$-flow sparsifier on $\log U$ unweighted graphs, where we use the fact that $U \leq \poly(n).$ To construct this unweighted flow sparsifier, we apply Lemma~\ref{lem:mainlemflowsparsify} iteratively as follows.
Starting with the instance unweighted graph $G_1=(V,E_1),$ we run the algorithm of Lemma~\ref{lem:mainlemflowsparsify} on the current graph $G_t = (V,E_t)$  to produce the sets $F_t$ and $F'_t,$ the weight vector $w_{F'_t}$ and the embedding $\mvar{H_t}:\R^{F'_t} \rightarrow R^E.$ To proceed to the next iteration, we then define $E_{t+1} \defeq E_t \setminus F_t$ and move on to $G_{t+1}.$

By Lemma~\ref{lem:mainlemflowsparsify}, at every iteration $t,$ $|F_t| \geq \frac{1}{2} \cdot |E_t|,$ so that $|E_{t+1}| \leq \frac{1}{2} \cdot |E_t|.$
This shows that there can be at most $T \leq \log(|E_1|) = O(\log n)$ iterations.

After the last iteration $T,$ we have effectively partitioned $E_1$ into disjoint subsets $\{F_t\}_{t \in [T]},$ where each $F_t$ is well-approximated but the weighted edgeset $F'_t.$  We then output the weighted graph $G'=(V, E'\defeq \cap_{t=1}^T F'_t, w' \defeq \sum_{t=1}^{T} w_{F'_t}),$ which is the sum of the disjoint weighted edges sets $\{F'_t\}_{t \in [T]}.$ We also output the embedding $\mEmbed: \R^{E'} \rightarrow \R^{E}$ from $G'$ to $G,$ defined as the direct sum
$$
\mEmbed = \bigoplus_{t=1}^T \mvar{H_t}.
$$
In words, $\mEmbed$ maps an edge $e' \in E'$ by finding $t$ for which $e' \in F'_t$ and applying the corresponding $\mvar{H_t}.$

We are now ready to prove that this algorithm with output $G'$ and $\mvar{M}$ is an efficient $(\otilde(n), \epsilon, \otilde(n))$-flow sparsifier. 
To bound the capacity ratio $U'$ of $G'$, we notice that
$$
U' \leq \max_t \frac{\max_{e \in F'_t} w_{F'_t}(e)}{\min_{e \in F'_t} w_{F'_t}(e)} \leq \poly(n),
$$
where we used the fact that the sets $F'_t$ are disjoint and the guarantee on the range of $w_{F'_t}.$

Next, we bound the sparsity of $G'.$ By Lemma~\ref{lem:mainlemflowsparsify}, $F'_t$ contains at most $\otilde(n)$ edges. As a result, we get the required bound
$$
|E'| = \sum_{t=1}^T |F'_t| \leq \otilde(Tn) = \otilde(n).
$$
For the cut approximation, we consider any $S \subseteq V.$ By the cut guarantee of Lemma~\ref{lem:mainlemflowsparsify}, we have that, for all $t \in [T],$
$$
(1-\epsilon) |\cutset_G(S) \cap F_t|\leq w_{F'_t}(\cutset_G(S) \cap F'_t) \leq (1+\epsilon) |\cutset_G(S) \cap F_t|.
$$
Summing over all $t,$ as $E' = \mathop{\dot{\bigcup}} F'_t$ and $E = \mathop{\dot{\bigcup}}F_t,$ we obtain the required approximation
$$
(1-\epsilon) |\cutset_G(S)|\leq w'(\cutset_{G'}(S)) \leq (1+\epsilon) |\cutset_G(S)|.
$$
The congestion of $\mvar{M}$ can be bounded as follows
$$
\congest(\mvar{M}) \leq \sum_{t=1}^T \congest(\mvar{H_t}) = \otilde(T) = \otilde(1).
$$
To conclude the proof, we address the efficiency of the flow sparsifier.
The algorithm applies the routine of Lemma~\ref{lem:mainlemflowsparsify} for $T=\otilde(1)$ times and hence runs in time $\otilde(m),$ as required.
Invoking the embedding $\mvar{M}$ requires invoking each of the $T$ embeddings ${\mvar{H_t}}.$ This takes time $\otilde(Tm)=\otilde(m).$

\end{proof}

\subsubsection{Flow Sparsification of Unweighted Graphs: Proof of Lemma~\ref{lem:mainlemflowsparsify}}

In this subsection, we prove Lemma~\ref{lem:mainlemflowsparsify}. Our starting point is the following decomposition statement, which shows that we can form a partition of an unweighted graph where most edges do not cross the boundaries and the subgraphs induced within each set of this partition are near-expanders.
The following lemma is implicit in Spielman and Teng's local clustering approach to spectral sparsification~\cite{STspectralSparse}. 
\begin{lemma}[Decomposition Lemma]\label{lem:decomp_lemma}
For an unweighted graph $G=(V,E),$ in $\tilde{O}(m)$-time we can produce a partition $V_1, \ldots, V_k$ of $V$ and a collection of sets $S_1, \ldots, S_k \subseteq V$ with the following properties:
\begin{itemize}
\item For all $i,$ $S_i$ is contained in $V_i.$
\item For all $i,$ there exists a set $T_i$ with $S_i \subseteq T_i \subseteq V_i,$ such that
$$
\conductance(G(T_i)) \geq \Omega\left(\frac{1}{\log^2 n}\right).
$$
\item At least half of the edges are found within the sets $\{S_i\},$ i.e.
$$
\sum^k_{i=1}|E(S_i)| = \sum_{i=1}^k |\{e=\{a,b\} : a \in S_i, b \in S_i\}| \geq \frac{1}{2} |E|. 
$$
\end{itemize} 
\end{lemma}

To design an algorithm satisfying the requirements of Lemma~\ref{lem:mainlemflowsparsify}, we start by appling the Decomposition Lemma to our unweighted input graph $G=(V,E)$ to obtain the partition $\{V_i\}_{i \in [k]}$ and the sets $\{S_i\}_{i \in [k]}.$ We let $G_i \defeq (S_i, E(S_i)).$ To reduce the number of edges, while preseving cuts, we apply a spectral sparsification algorithm to each $G_i.$ Concretely, by applying the spectral sparsification by effective resistances of Spielman and Srivastava~\cite{Spielman:2008:GSE:1374376.1374456} to each $G_i,$ we obtain weighted graphs $G'_i = (S_i, E'_i \subseteq E(S_i), w'_i)$ in time $\sum_{i=1}^k \otilde(|E(S_i)|) \leq \otilde(|E|)$ with $|E'_i| \leq \otilde(|S_i|)$ and the property that cuts are preserved\footnote{The spectral sparsification result actually yields the stronger spectral approximation guarantee, but for our purposes the cut guarantee suffices.} for all $i$:
$$
\forall S \subseteq S_i \enspace , \enspace (1-\epsilon) \cdot |\cutset_{G_i}(S)| \leq w'_i(\cutset_{G'_i}(S)) \leq (1+\epsilon) \cdot |\cutset_{G_i}(S)|.
$$
Moreover, the spectral sparsification of~\cite{Spielman:2008:GSE:1374376.1374456} constructs the weights $\{w'_i(e)\}_{e \in E'_i}$ such that 
$$
\forall e \in E'_i \enspace, \enspace  \frac{1}{\poly(n)} \leq \frac{1}{\poly(|S_i|)} \leq w'_i(e) \leq |S_i| \leq n.
$$
To complete the description of the algorithm, we output the partition $(F, \bar{F})$ of $E,$ where
$$
F \defeq \bigcup_{i=1}^k E(S_i).
$$
We also output the set of weighted sparsified edges $F'.$ 
$$
F' \defeq \bigcup_{i=1}^k E'_i.
$$
The weight $w_{F'}(e)$ of edge $e \in F'$ is given by finding $i$ such that $e \in E'_i$ and setting $w_{F'}(e) = w'_i(e).$

We now depart from Spielman and Teng's construction by endowing our $F'$ with an embedding onto $G.$ The embedding $\mvar{H}: \R^{F'} \rightarrow \R^E$ of the graph $H=(V, F', w_{F'})$ to $G$ is constructed by using the oblivious electrical-flow routing of $E(S_i)$ into $G(V_i).$
More specifically, as the sets $\{V_i\}$ partition $V,$ the embedding $\mvar{H}$ can be expressed as the following direct sum over the orthogonal subspaces $\{\R^{E(V_i) \times E'_i}\}.$
$$
\mvar{H} \defeq \left(
\bigoplus_{i=1}^k \incMatrix_{E(V_i)} \lapPseudo_{G(V_i)} 
\incMatrix^T_{E(V_i)} \iMatrix_{(E(V_i), E'_i)}
\right),
$$
where $\iMatrix_{(E(V_i), E'_i)}$ is the identity mapping of the edges $E'_i \subseteq E(V_i)$ of $F'$  over $V_i$ to the edges $E(V_i)$ of $V_i$ in $G.$
Notice that there is no dependence on the resistances over $G$ as $G$ is unweighted.

This complete the description of the algorithm. We are now ready to give the proof of Lemma~\ref{lem:mainlemflowsparsify}.

\begin{proof}[Proof of Lemma~\ref{lem:mainlemflowsparsify}]
The algorithm described above performs a decomposition of the input graph $G=(V,E)$ in time $\otilde(m)$ by the Decomposition Lemma. By the result of Spielman and Srivastava~\cite{Spielman:2008:GSE:1374376.1374456}, each $G_i$ is sparsified in time $\otilde(|E(S_i)|).$ Hence, the sparsification step requires time $\otilde(m)$ as well. This shows that the algorithm runs in $\otilde(m)$-time, as required.

By the Decomposition Lemma, we know that $|F| = \sum_{i=1}^k |E(S_i)| \geq \frac{|E|}{2},$ which satisfies the requirement of the Lemma. Moreover, by the spectral sparsification result, we know that $|F'| = \sum_{i=1}^k |E'_i| \leq \sum_{i=1}^k \otilde(|S_i|) \leq \otilde(n),$ as required. We also saw that by construction the weights $w_{F'}$ are bounded:
$$
\forall e \in F' \enspace ,  \enspace \frac{1}{\poly(n)} \leq w_{F'}(e) \leq n.
$$

To obtain the cut-approximation guarantee, we use the fact that for every $i,$ by spectral sparsification,
$$
\forall S \subseteq S_i \enspace , \enspace (1-\epsilon) \cdot |\cutset_{G_i}(S)| \leq w'_i(\cutset_{G'_i}(S)) \leq (1+\epsilon) \cdot |\cutset_{G_i}(S)|.
$$
We have $H'=(V,F', w_{F'})$ and $H=(V,F).$
Consider now $T \subseteq V$ and apply the previous bound to $T \cap S_i$ for all $i.$ Because $F' \subseteq F = \cup_{i=1}^k E(S_i),$ we have that summing over the $k$ bounds yields
$$
\forall T \subseteq V \enspace , \enspace (1 - \epsilon) |\cutset_H(T)| \leq w_{F'}(\cutset_{H'}(T)) \leq (1+ \epsilon) |\cutset_H(T)|,
$$
which is the desired cut-approximaton guarantee.

Finally, we are left to prove that the embedding $\mvar{H}$ from $H'=(V,F', w_{F'})$ to $G=(V,E)$ has low congestion and can be applied efficiently.
By definition of congestion,
\[
\congest(\mvar{H}) =
\max_{\vx \in \R^{F'}} \frac{\normInf{{\mvar{H} \vx}}}{\normInf{\capacityMatrix_{F'}^{-1} \vx}}
= \normInf{|\mvar{H}| \capacityMatrix_{F'} \onesVec_{F'}} =
\left\|\left|
\bigoplus_{i=1}^k \incMatrix_{E(V_i)} \lapPseudo_{G(V_i)} 
\incMatrix^T_{E(V_i)} \iMatrix_{(E(V_i), E'_i)}
\right| \capacityMatrix_{F'} \onesVec_{F'}\right\|_\infty.
\]
Decomposing $\R^{E}$ into the subspaces $\{\R^{E(V_i)}\}$ and $\R^{F'}$ into the subspaces $\{\R^{E'_i}\}$ we have:
$$
\congest(\mvar{H}) \leq \max_{i \in [k]} \left\|\left|
\incMatrix_{E(V_i)} \lapPseudo_{G(V_i)} 
\incMatrix^T_{E(V_i)} \iMatrix_{(E(V_i), E'_i)} \right| \capacityMatrix_{E'_i} \onesVec_{E'_i} \right\|_\infty.
$$
For each $i \in [k],$ consider now the set of demands $D_i$ over $V_i$, $D_i \defeq \{\demands_e\}_{e \in E'_i},$ given by the edges of $E'_i$ with their capacities $w'_i.$ That is, $\demands_e \in \R^{V_i}$ is the demand corresponding to edge $e \in E'_i$ with weight $w'_i(e).$ Consider also the electrical routing  ${\mElRout}_{,i} = \incMatrix_{E(V_i)} \lapPseudo_{G(V_i)}$ over $G(V_i).$
Then:
$$
\congest(\mvar{H}) \leq \max_{i \in [k]} \congest({\mElRout}_{,i} D_i)
$$
Notice that, by construction, $D_i$ is routable in $G'_i=(S_i, E'_i, w'_i)$ and $\opt_{G'_i}(D_i) = 1.$ 
But, by our use of spectral sparsifiers in the construction, $G'_i$ is an $\epsilon$-cut approximation of $G_i.$ Hence, by the flow-cut gap of Aumann and Rabani~\cite{flowcutgap}, we have:
$$
\opt_{G_i}(D_i) \leq O(\log(|D_i|)) \cdot  \opt_{G'_i}(D_i) \leq \otilde(1).
$$
When we route $D_i$ oblivious in $G(V_i)$, we can consider the $E(S_i)$-competitive ratio $\rho^{E(S_i)}({\mElRout}_{,i})$ of the electrical routing  ${\mElRout}_{,i} = \incMatrix_{E(V_i)} \lapPseudo_{G(V_i)},$ as $D_i$ is routable in $E(S_i),$ because $E'_i \subseteq E(S_i).$ 
We have
$$
\congest(\mvar{H})  \leq \max_{i \in [k]} \rho^{E(S_i)}_{G(V_i)}({\mElRout}_{,i}) \cdot \opt^{E(S_i)}_{G(V_i)}(D_i) = \max_{i \in [k]} \rho^{E(S_i)}_{G(V_i)}({\mElRout}_{,i}) \cdot \opt_{G_i}(D_i),
$$
Finally, putting these bounds together, we have:
$$
\congest(\mvar{H}) \leq \max_{i \in [k]} \rho^{E(S_i)}_{G(V_i)}({\mElRout}_{,i}) \cdot \opt_{G_i}(D_i) \leq \otilde(1) \cdot \max_{i \in [k]} \rho^{E(S_i)}_{G(V_i)}({\mElRout}_{,i}).
$$
But, by the Decomposition Lemma, there exists $T_i$ with $S_i \subseteq T_i \subseteq V_i$ such that
$$
\conductance(G(T_i)) \geq \Omega \left(\frac{1}{\log^2 n}\right).
$$
Then, by Lemma~\ref{lem:elecsubroutinggoal}, we have that:
$$
\rho^{E(S_i)}_{G(V_i)}({\mElRout}_{,i}) \leq O\left(\frac{\log\vol(G(T_i))}{\conductance(G(T_i))^2}\right) \leq \otilde(1).
$$
This concludes the proof that $\congest(\mvar{H}) \leq \otilde(1).$
To complete the proof of the Lemma, we just notice that $\mvar{H}$ can be invoked in time $\otilde(m).$ A call of $\mvar{H}$ involves solving $k$-electrical-problems, one for each $G(V_i).$ This can be done in time $\sum_{i=1}^k \otilde(|E(V_i)|) \leq \otilde(m),$ using any of the nearly-linear Laplacian system solvers available, such as \cite{koszSolver}.
\end{proof}

\section{Removing Vertices in Oblivious Routing Construction}
\label{sec:less_vertices}

In this section we show how to reduce computing an efficient oblivious routing on a graph $G = (V, E)$ to computing an oblivious routing for $t$ graphs with $\otilde(\frac{|V|}{t})$ vertices and at most $|E|$ edges. Formally we show

\begin{theorem}[Node Reduction \emph{(Restatement)}]\label{thm:node_red_restatement}
Let $G = (V, E, \capacityVec)$ be an undirected capacitated graph with capacity ratio $U$. For all $t > 0$ in $\otilde(t \cdot |E|)$ time we can compute graphs $G_1, \ldots , G_t$ each with at most $\otilde(\frac{|E| \log(U)}{t})$ vertices, at most $|E|$ edges, and capacity ratio at most $|V| \cdot U$, such that given oblivious routings $\mRoute_i$ for each $G_i$, in $\otilde(t \cdot |E|)$ time we can compute an oblivious routing $\mRoute \in \R^{E \times V}$ on $G$ such that
\[
\timeOf{\mRoute} = \otilde\left(t \cdot |E| + \sum_{i = 1}^{t} \timeOf{\mRoute_i}\right)
\enspace \text{ and } \enspace
\rho(\mRoute) = \otilde\left(\max_{i} \rho(\mRoute_i)\right)
\]
\end{theorem}

We break this proof into several parts. First we show how to embed $G$ into a collection of $t$ graphs consisting of trees minus some edges which we call \emph{patrial tree embeddings} (\sectionref{sec:from_graph_to_partialtreee}). Then we show how to embed a partial tree embedding in an ``almost $j$-tree'' \cite{Madry10}, that is a graph consisting of a tree and a subgraph on at most $j$ vertices, for $j = 2t$ (\sectionref{sec:from_patrial_tree_to_almost_tree}). Finally, we show how to reduce oblivious routing on an almost $j$-tree to oblivious routing on a graph with at most $O(j)$ vertices by removing degree-1 and degree-2 vertices (\sectionref{sec:from_almost_tree_to_less_vert}). Finally, in \sectionref{sec:route_final} we put this all together to prove \theoremref{thm:node_reduction}. 

We remark that much of the ideas in the section were either highly influenced from \cite{Madry10} or are direct restatements of theorems from \cite{Madry10} adapted to our setting. We encourage the reader to look over that paper for further details regarding the techniques used in this section.

\subsection{From Graphs to Partial Tree Embeddings}
\label{sec:from_graph_to_partialtreee}

To prove Theorem \ref{thm:node_reduction}, we make heavy use of spanning trees and various properties of them. In particular, we use the facts that for every pair of vertices there is a unique \emph{tree path} connecting them, that every edge in the tree \emph{induces a cut} in the graph, and that we can embed a graph in a tree by simply routing ever edge over its tree path and that the congestion of this embedding will be determined by the \emph{load} the edges place on tree edges. We define these quantities formally below.

\begin{definition}[Tree Path]\label{def:tree_path}
For undirected graph $G = (V, E)$, spanning tree $\tree$, and all $a, b \in V$ we let $\treePath{a, b} \subseteq E$ denote the unique path from $a$ to $b$ using only edges in $\tree$ and we let $\treePathVec{a,b} \in \Redgevec$ denote the vector representation of this path corresponding to the unique vector sending one one unit from $a$ to $b$ that is nonzero only on $\tree$ (i.e. $\incMatrix^T \treePathVec{a,b} = \demands_{a,b}$ and $\forall e \in \offtreeEdgeSet$ we have $\treePathVec{a,b}(e) = 0$)
\end{definition}

\begin{definition}[Tree Cuts]\label{def:tree_cuts}
For undirected $G = (V, E)$ and spanning tree $\tree \subseteq E$ the \emph{edges cut by $e$}, $\cutset_\tree(F)$, and the edges cut by $F$, $\cutset_\tree(e)$, are given by
\[
\cutset_\tree(e) \defeq \{e' \in E ~ | ~ e' \in \treePath{e}\}
\enspace \text{ and } \enspace
\cutset_\tree(F) \defeq \cup_{e \in F} \cutset(e)
\]
\end{definition}

\begin{definition}[Tree Load]\label{def:tree_load}
For undirected capacitated $G = (V, E, \capacityVec)$ and spanning tree $\tree \subseteq E$ the load on edge $e \in E$ by $\tree$, $\congest_\tree(e)$ is given by
$
\load_\tree(e) = \sum_{e' \in E | e \in \treePath{e'}} \capacityVec_{e'}
$
\end{definition}

While these properties do highlight the fact that we could just embed our graph into a collection of trees to simplify the structure of our graph, this approach suffers from a high computational cost \cite{Racke:2008:OHD:1374376.1374415}. Instead we show that we can embed parts of the graph onto collections of trees at a lower computational cost but higher complexity. In particular we will consider what we call partial tree embeddings.

\begin{definition}[Partial Tree Embedding \footnote{This is a restatement of the $H(T, F)$ graphs in \cite{Madry10}.}]\label{def:partial_tree}
For undirected capacititated graph $G = (V, E, \capacityVec)$, spanning tree $\tree$ and spanning tree subset $F \subseteq \tree$ we define the \emph{partial tree embedding graph} $H = H(G, \tree, F) = (V, E', \capacityVec')$ to a be a graph on the same vertex set where
$
E' = \tree \cup \cutset_\tree(F)
$
and
\[
\forall e \in E'
\enspace : \enspace
\capacityVec'(e) = 
\begin{cases}
\load_{\tree}(e) & \text{if $e \in \tree \setminus F$}. \\
\capacityVec(e) & \text{otherwise}
\end{cases}
\]
Furthermore, we let $\mvar{M}_H \in \R^{E' \times E}$ denote the embedding from $G$ to $H(G, \tree, F)$ where edges not cut by $F$ are routed over the tree and other edges are mapped to themselves.
\[
\forall e \in E
\enspace : \enspace
\mEmbed_H(e) =
\begin{cases}
\treePathVec{e} & e \notin \cutset_\tree(F)\\
\indicVec{e} & \text{otherwise} \\
\end{cases}
\]
and we let $\mEmbed'_H \in \R^{E \times E'}$ denote the embeding from $H$ to $G$ that simply maps edges in $H$ to their corresponding edges in $G$, i.e. $\forall e \in E'$, $\mEmbed'_H(e) = \indicVec{e}$.
\end{definition}

Note that by definition $\congest(\mEmbed_H) \leq 1$, i.e. a graph embeds into its partial tree embedding with no congestion. However, to get embedding guarantees in the other direction more work is required. For this purpose we use a lemma from Madry \cite{Madry10} saying that we can construct a convex combination or a distribution of partial tree embeddings we can get such a guarantee. 

\begin{lemma}[Probabilistic Partial Tree Embedding \footnote{This in an adaptation of Corollary 5.6 in \cite{Madry10}}]
\label{lem:probabilistic_paritial_tree_embedding}
For any undirected capacitated graph $G = (V, E, \capacityVec)$ and any $t > 0$ in $\otilde(t \cdot m)$ time we can find a collection of partial tree embeddings $H_1 = H(G, \tree_1, F_1), \ldots, H_t = H(G, \tree_t, F_t)$ and coefficients $\lambda_i \geq 0$ with $\sum_i \lambda_i = 1$ such that $\forall i \in [t]$ we have $|F_i| = \otilde(\frac{m \log U}{t})$ and such that $\sum_i \lambda_i \mEmbed'_{H_i}$ embeds $G' = \sum_i \lambda_i G_i$ into $G$ with congestion $\otilde(1)$.
\end{lemma}

Using this lemma, we can prove that we can reduce constructing an oblivious routing for a graph to constructing oblivious routings on several partial tree embeddings.

\begin{lemma}
\label{lem:obliv_route_graph_to_partial}
Let the $H_i$ be graphs produced by \lemmaref{lem:probabilistic_paritial_tree_embedding} and for all $i$ let $\mRoute_i$ be an oblivious routing algorithm for $H_i$. It follows that $\mRoute = \sum_{i} \lambda_i \mEmbed'_{H_i} \mRoute_i$ is an oblivious routing on $G$ with $\rho(\mRoute) \leq \otilde(\max_{i} \rho(\mRoute_i) \log n)$ and $\timeOf{\mRoute} = O(\sum_{i} \timeOf{\mRoute_i})$
\end{lemma}

\begin{proof}
The proof is similar to the proof of \lemmaref{lem:embedding_lemma}. For all $i$ let $\capacityMatrix_i$ denote the capacity matrix of graph $G_i$. Then using \lemmaref{lem:equivalence_of_competetivity_and_norm} we get
\[
\rho(\mRoute)
= \normInf{\capacityMatrix^{-1} \mRoute \incMatrix^T \capacityMatrix}
= \left\|\sum_{i = 1}^{t} \lambda_i \capacityMatrix^{-1} \mEmbed'_{H_i} \mRoute_i \incMatrix^T \capacityMatrix \right\|_{\infty}
\]
Using that $\mEmbed_{H_i}$ is an embedding and therefore $\incMatrix^T_{H_i} \mEmbed_{H_i} = \incMatrix^T$ we get
\[
\rho(\mRoute)
= 
\left\|\sum_{i = 1}^{t} \lambda_i \capacityMatrix^{-1} \mEmbed'_{H_i} \mRoute_i \incMatrix^T_{H_i} \mEmbed_{H_i} \capacityMatrix 
\right\|_{\infty}
\leq 
\max_{j, k} \left\|\sum_{i = 1}^{t} \lambda_i \capacityMatrix^{-1} \mEmbed'_i \capacityMatrix_i  \right\|_{\infty} \cdot \rho(\mRoute_j) \cdot \congest(\mEmbed_{H_k})
\]
Since $\sum_{i} \lambda_i \mEmbed_{H_i}'$ is an embedding of congestion of at most $\otilde(1)$ and $\congest(\mEmbed_{H_k}) \leq 1$ we have the desired result.
\end{proof}

\subsection{From Partial Tree Embeddings To Almost-j-trees}
\label{sec:from_patrial_tree_to_almost_tree}

Here we show how to reduce constructing an oblivious routing for a partial tree embedding to constructing an oblivious routing for what Madry \cite{Madry10} calls an ``almost $j$-tree,'' the union of a tree plus a subgraph on at most $j$ vertices. First we define such objects and then we prove the reduction.

\begin{definition}[Almost $j$-tree]\label{def:almost_j}
We call a graph $G = (V, E)$ an \emph{almost $j$-tree} if there is a spanning tree $\tree \subseteq E$ such that the endpoints of $E \setminus \tree$ include at most $j$ vertices. 
\end{definition}

\begin{lemma}
\label{lem:partialtree_to_jtree}
For undirected capacitated $G = (V, E, \capacityVec)$ and partial tree embedding $H = H(G, \tree, F)$ in $\otilde(|E|)$ time we can construct an almost $2 \cdot |F|$-tree $G' = (V, E', \capacityVec')$ with $|E'| \leq |E|$ and an embedding $\mEmbed'$ from $G'$ to $H$ such that $H$ is embeddable into $G'$ with congestion 2, $\congest(\mEmbed') = 2$, and $\timeOf{\mEmbed'} = \otilde(|E|)$.
\end{lemma}

\begin{proof}
For every $e = (a, b) \in E$, we let $\first(e) \in V$ denote the first vertex on tree path $\treePath{(a,b)}$ incident to $F$ and we let $\last(e) \in V$ denote the last vertex incident to $F$ on tree path $\treePath{(a,b)}$. Note that for every $e = (a, b) \in \tree$ we have that $(\first(e), \last(e)) = e$.

We define $G' = (V, E', \capacityVec')$ to simply be the graph that consists of all these $(\first(e), \last(e))$ pairs
\[
E' = \{(a, b) ~ | ~ \exists e \in E \text{ such that } (a, b) = (\first(e), \last(e))\}
\]
and we define the weights to simply be the sums
\[
\forall e' \in E'
\enspace : \enspace
\capacityVec'(e') \defeq
\sum_{e \in E ~ | ~ e = (\first(e'), \last(e'))}
\capacityVec(e)
\]
Now to embed $H$ in $G'$ we define $\mEmbed$ by
\[
\forall e = (a, b) \in E
\enspace : \enspace
\mEmbed \indicVec{e} 
= \treePathVec{a, \first(e)} 
+ \indicVec{(\first(e), \last(e))}
+ \treePathVec{\last(e), b}
\]
and to embed $G'$ in $H$ we define $\mEmbed'$ by
\[
\forall e' \in E
\enspace : \enspace
\mEmbed' \indicVec{e'} 
= 
\sum_{e = (a, b) \in E ~ | ~ e' = (\first(e), \last(e))}
\frac{\capacityVec(e)}{\capacityVec'(e')}
\left[
\treePathVec{\first(e),a} + \indicVec{(a,b)} + \treePathVec{b,\last(e)}
\right]
\]
In other words we route edges in $H$ along the tree until we encounter nodes in $F$ and then we route them along added edges and we simply route the other way for the reverse embedding. By construction clearly the congestion of the embedding in either direction is 2.

To bound the running time, we note that by having every edge $e$ in $H$ maintain its $\first(e)$ and $\last(e)$ information, having every edge $e'$ in $E'$ maintain the set $\{e \in E | e' = (\first(e), \last(e))\}$ in a list, and using link cut trees \cite{Sleator:1981:DSD:800076.802464} or the static tree structure in \cite{koszSolver} to update information along tree paths we can obtain the desired value of $\timeOf{\mEmbed'}$.
\end{proof}

\subsection{From Almost-J Trees to Less Vertices}
\label{sec:from_almost_tree_to_less_vert}

Here we show that by ``greedy elimination'' \cite{Spielman:Solver:DBLP:journals/corr/abs-cs-0607105} \cite{Koutis:2010:AOS:1917827.1918388} \cite{Koutis:2011:NLN:2082752.2082901}, i.e. removing all degree 1 and degree 2 vertices in $O(m)$ time we can reduce oblivious routing in almost-$j$-trees to oblivious routing in graphs with $O(j)$ vertices while only losing $O(1)$ in the competitive ratio. Again, we remark that the lemmas in this section are derived heavily from \cite{Madry10} but repeated for completeness and to prove additional properties that we will need for our purposes.

We start by showing that an almost-$j$-tree with no degree 1 or degree 2 vertices has at most $O(j)$ vertices.

\begin{lemma}
\label{lem:jtree_largedeg_bound}
For any almost $j$-tree $G = (V, E)$ with no degree 1 or degree 2 vertices, we have $|V| \leq 3j - 2$.
\end{lemma}

\begin{proof}
Since $G$ is an almost $j$-tree, there is some $J \subseteq V$ with $|J| \leq j$ such that the removal of all edges with both endpoints in $J$ creates a forest. Now, since $K = V - J$ is incident only to forest edges clearly the sum of the degrees of the vertices in $K$ is at most $2(|V| - 1)$ (otherwise there would be a cycle). However, since the minimum degree in $G$ is 3, clearly this sum is at least $3(|V| - j)$. Combining yields that $3|V| - 3j \leq 2|V| - 2$.
\end{proof}

Next, we show how to remove degree one vertices efficiently.

\begin{lemma}[Removing Degree One Vertices]
\label{lem:remove_deg_one_vert}
Let $G = (V, E, \capacityVec)$ be an unweighted capacitated graph, let $a \in V$ be a degree 1 vertex, let $e = (a, b) \in E$ be the single edge incident to $a$, and let $G' = (V', E', \capacityVec')$ be the graph that results from simply removing $e$ and $a$, i.e. $V' = V \setminus \{a\}$ and $E' = E \setminus \{e\}$. 
Given $a \in V$ and an oblivious routing algorithm $\mRoute'$ in $G'$ in $O(1)$ time we can construct an oblivious routing algorithm $\mRoute$ in $G$ such that
\[
\timeOf{\mRoute} = O(\timeOf{\mRoute'} + 1)
\enspace \text{, and} \enspace
\rho(\mRoute) = \rho(\mRoute')
\]
\end{lemma}

\begin{proof}
For any demand vector $\demands$, the only way to route demand at $a$ in $G$ is over $e$. Therefore, if $\incMatrix \flow = \demands$ then $\flow(e) = \demands$. Therefore, to get an oblivious routing algorithm on $G$, we can simply send demand at $a$ over edge $e$, modify the demand at $b$ accordingly, and then run the oblivious routing algorithm on $G'$ on the remaining vertices. The routing algorithm we get is the following
\[
\mRoute \defeq
\iMatrix_{E' \rightarrow E} \mRoute' (\iMatrix + \indicVec{b} \indicVec{a}^T)
+ \indicVec{e}\indicVec{a}^T
\]
Since all routing algorithms send this flow on $e$ we get that $\rho(\mRoute) = \rho(\mRoute')$ and since the above operators not counting $\mRoute$ have only $O(1)$ entries that are not the identity we can clearly implement the operations in the desired running time.
\end{proof}

Using the above lemma we show how to remove all degree $1$ and $2$ vertices in $O(m)$ time while only increasing the congestion by $O(1)$.

\begin{lemma}[Greedy Elimination]
\label{lem:greedy_elimination}
Let $G = (V, E, \capacityVec)$ be an unweighted capacitated graph and let $G' = (V', E', \capacityVec')$ be the graph the results from iteratively removing vertices of degree $1$ and replacing degree $2$ vertices with an edge connecting its neighbors of the minimum capacity of its adjacent edges. We can construct $G'$ in $O(m)$ time and given an oblivious routing algorithm $\mRoute'$ in $G'$ in $O(1)$ time we can construct an oblivious routing algorithm $\mRoute$ in $G$ such that \footnote{Note that the constant of 4 below is improved to 3 in \cite{Madry10}.}
\[
\timeOf{\mRoute} = O(\timeOf{\mRoute'} + |E|)
\enspace \text{, and} \enspace
\rho(\mRoute) \leq 4 \cdot \rho(\mRoute')
\]
\end{lemma}

\begin{proof}
First we repeatedly apply Lemma \ref{lem:remove_deg_one_vert} repeatedly to in reduce to the case that there are no degree 1 vertices. By simply array of the degrees of every vertex and a list of degree 1 vertices this can be done in $O(m) $ time. We denote the result of these operations by graph $K$.

Next, we repeatedly find degree two vertices that have not been explored and explore this vertices neighbors to get a path of vertices, $a_1, a_2, \ldots , a_k \in V$ for $k > 3$ such that each vertex $a_2, \ldots, a_{k - 1}$ is of degree two. We then compute $j = \argmin_{i \in [k - 1]} \capacityVec(a_i, a_{i + 1})$, remove edge $(a_j, a_{j + 1})$ and add an edge $(a_1, a_k)$ of capacity $\capacityVec(a_j, a_{j + 1})$. We denote the result of doing this for all degree two vertices by $K'$ and note that again by careful implementation this can be performed in $O(m)$ time.

Note that clearly $K$ is embeddable in $K'$ with congestion 2 just by routing every edge over itself except the removed edges which we route by the path plus the added edges. Furthermore, $K'$ is embeddable in $K$ with congestion 2 again by routing every edge on itself except for the edges which we added which we route back over the paths they came from. Furthermore, we note that clearly this embedding and the transpose of this operator is computable in $O(m)$ time.

Finally, by again repeatedly applying \lemmaref{lem:remove_deg_one_vert} to $K'$ until there are no degree 1 vertices we get a graph $G'$ that has no degree one or degree two vertices (since nothing decreased the degree of vertices with degree more than two). Furthermore, by \lemmaref{lem:remove_deg_one_vert} and by \lemmaref{lem:embedding_lemma} we see that we can compose these operators to compute $\mRoute$ with the desired properties.
\end{proof}

\subsection{Putting It All Together}
\label{sec:route_final}

Here we put together the previous components to prove the main theorem of this section.

\begin{proof}[Node Reduction \theoremref{thm:node_reduction}]
Using \lemmaref{lem:probabilistic_paritial_tree_embedding} we can construct $G' = \sum_{i = 1}^{t} \lambda_i G_i$ and embeddings $\mEmbed_1, \ldots, \mEmbed_t$ from $G_i$ to $G$. Next we can apply Lemma \ref{lem:partialtree_to_jtree} to each $G_i$ to get almost-$j$-trees  $G'_1, \ldots , G'_t$ and embeddings $\mEmbed'_1, \ldots, \mEmbed'_t$ from $G'_i$ to $G_i$. Furthermore, using Lemma \ref{lem:greedy_elimination} we can construction graphs $G''_1, \ldots, G''_t$ with the desired properties (the congestion ratio property follows from the fact that we only add capacities during these reductions)

Now given oblivious routing algorithms $\mRoute''_1, \ldots, \mRoute''_t$ on the $G''_i$ and again by Lemma \ref{lem:greedy_elimination} we could get oblivious routing algorithms $\mRoute'_1, \ldots , \mRoute'_t$ on the $G'_i$ with constant times more congestion. Finally, by the guarantees of Lemma \ref{lem:embedding_lemma} we have that $\mRoute \defeq \sum_{i = 1}^{t} \lambda \mEmbed_i \mEmbed'_i \mRoute'_i$ is an oblivious routing algorithm that satisfies the requirements.
\end{proof}

\section{Nonlinear Projection and Maximum Concurrent Flow}
\label{sec:nonlinear_projection}

\subsection{Gradient Descent Method for Nonlinear Projection Problem}

In this section, we strengthen and generalize the \textbf{MaxFlow
}algorithm to a more general setting. We believe this algorithm may
be of independent interest as it includes maximum concurrent
flow problem, the compressive sensing problem, etc. For some norms,
e.g. $\norm{\cdot}_{1}$ as typically of interest compressive sensing, the Nesterov algorithm \cite{nesterov2005smooth} can be used to replace gradient
descent. However, this kind of accelerated method is not known in the general norm settings as good proxy function may not exist at all. Even worse,
in the non-smooth regime, the minimization problem on the $\norm{\cdot}_{p}$ space
with $p>2$ is difficult under some oracle assumption \cite{nemirovsky1983problem}.
For these reasons we focus here on the gradient descent method which is always
applicable.

Given a norm $\norm{\cdot}$, we wish to solve the what we call the \emph{non-linear projection} problem
\[
\min_{\vx\in L}\norm{\vx-\vy}
\]
where $\vy$ is an given point and $L$ is a linear subspace. We assume
the following:

\begin{assumption} \label{ass:nonlinear_projection}$\ $ 
\begin{enumerate}
\item There are a family of convex differentiable functions $f_{t}$ such
that for all $\vx\in L$, we have 
\[
\norm{\vx}\leq f_{t}(\vx)\leq\norm{\vx}+Kt
\]
and the Lipschitz constant of $\nabla f_{t}$ is $\frac{1}{t}$. 
\item There is a projection matrix $\mvar P$ onto the subspace $L$. 
\end{enumerate}
\end{assumption}

In other words we assume that there is a family of regularized objective functions $f_t$ and a projection matrix $\mvar P$, which we can think of as an approximation algorithm of this projection problem. 

Now, let $\vec{x}^{*}$ be a minimizer of $\min_{\vx\in L}\norm{\vx-\vy}$.
Since $\vec{x}^{*}\in L$, we have $\mvar P\vec{x}^{*}=\vec{x}^{*}$
and hence 
\begin{eqnarray}
\norm{\mvar P\vec{y}-\vec{y}} & \leq & \norm{\vy-\vec{x}^{*}}+\norm{\vec{x}^{*}-\mvar P\vy}\nonumber \\
 & \leq & \norm{\vy-\vec{x}^{*}}+\norm{\mvar P\vec{x}^{*}-\mvar P\vy}\nonumber \\
 & \leq & \left(1+\norm{\mvar P}\right)\min_{\vx\in L}\norm{\vx-\vy}.\label{eq:approximate_ratio}
\end{eqnarray}
Therefore, the approximation ratio of $\mvar P$ is $1 + \norm{\mvar P}$ and we see that our problem is to show that we can solve nonlinear projection using a decent linear projection matrix. Our algorithm for solving this problem is below.

\begin{center}
\begin{tabular}{|l|}
\hline 
\textbf{NonlinearProjection}\tabularnewline
\hline 
\hline 
Input: a point $\vy$ and $\text{OPT}=\min_{\vx\in L}\norm{\vx-\vy}$.\tabularnewline
\hline 
1. Let $\vec{y_{0}}=\left(\mvar I-\mvar P\right)\vec{y}$ and $\vec{x_{0}}=0$.\tabularnewline
\hline 
2. For $j=0,\cdots,$ until $2^{-j}\norm{\mvar P}\leq\frac{1}{2}$\tabularnewline
\hline 
3. $\quad$If $2^{-j}\norm{\mvar P}>1$, then let $t_{j}=\frac{2^{-(j+2)}\norm{\mvar P}\text{OPT}}{K}\text{ and }k_{j}=3200\norm{\mvar P}^{2}K$.\tabularnewline
\hline 
4. $\quad$If $2^{-j}\norm{\mvar P}\leq1$, then let $t_{j}=\frac{\varepsilon\text{OPT}}{2K}\text{ and }k_{j}=\frac{800\norm{\mvar P}^{2}K}{\varepsilon^{2}}$.\tabularnewline
\hline 
5. $\quad$Let $g_{j}(\vx)=f_{t_{j}}(\mvar P\vx-\vec{y_{j}})$ and
$\vec{x_{0}}=0$.\tabularnewline
\hline 
6. $\quad$For $i=0,\cdots,k_{j}-1$\tabularnewline
\hline 
7. $\quad\quad\vx_{i+1}=\vx_{i}-\frac{t}{\norm{\mvar P}^{2}}\dualVec{(\gradient g_j(\vx_{i}))}.$\tabularnewline
\hline 
8. $\quad$Let $\vec{y_{j+1}}=\vec{y_{j}}-\mvar P\vx_{k_{j}}$.\tabularnewline
\hline 
9. Output $\vec{y}-\vec{y}_{\text{last}}$.\tabularnewline
\hline 
\end{tabular}
\par\end{center}

Note that this algorithm and its proof are quite similar to \theoremref{thm:MaxFlowAlgorithm} but modified to scale parameters over an outer loop. By changing the parameter $t$ we can decrease the dependence of the initial error.\footnote{This is an idea that has been applied previously to solve linear programming problems \cite{nesterov2004rounding}.}

\begin{theorem} \label{thm:NonLinearProjection}Assume the conditions
in Assumption \ref{ass:nonlinear_projection} are satisfied. Let $\runtime$
be the time needed to compute $\mvar Px$ and $\mvar P^{T}x$ and
$\dualVec x$. Then, \textbf{NonlinearProjection} outputs a vector
$\vx$ with $\norm{\vx}\leq(1+\varepsilon)\min_{\vx\in L}\norm{\vx-\vy}$
and the algorithm takes time 
\[
O\left(\norm{\mvar P}^{2}K\left(\runtime+m\right)\left(\frac{1}{\varepsilon^{2}}+\log\norm{\mvar P}\right)\right).
\]
\end{theorem} \begin{proof} We prove by induction on $j$ that when $2^{-(j-1)}\norm{\mvar P}\geq1$ we have $\norm{\vec{y_{j}}}\leq\left(1+2^{-j}\norm{\mvar P}\right)\text{OPT}$.

For the base case ($j = 0$), (\ref{thm:NonLinearProjection}) shows that $\norm{\vec{y_{0}}}\leq\left(1+\norm{\mvar P}\right)\text{OPT}.$

For the inductive case we assume that the assertion holds for some $j$. We start by bounding the corresponding $R$ in \theoremref{thm:gradient_descent} for $g_{j}$, which we denote $R_{j}$. Note that 
\[
g_{j}(\vx_0)=f_{t_{j}}(-\vec{y_{j}})\leq\norm{\vec{y_{j}}}+Kt_{j}\leq\left(1+2^{-j}\norm{\mvar P}\right)\text{OPT}+Kt_{j}.
\]
Hence, the condition that $g_{j}(\vx)\leq g_{j}(\vx_{0})$ implies
that 
\[
\norm{\mvar P\vx-\vec{y_{j}}}\leq\left(1+2^{-j}\norm{\mvar P}\right)\text{OPT}+Kt_{j}.
\]
Take any $\vy\in X^{*}$, let $\vc=\vx-\mvar P\vx+\vy$, and note
that $\mvar P\vc=\mvar P\vy$ and therefore $\vc\in X^{*}$. Using
these facts, we can bound $R_{j}$ as follows 
\begin{align*}
R_{j} & =\max_{\vx\in\redgevec~:~g_{j}(\vx)\leq g_{j}(\vx_{0})}\left\{ \min_{\vx^{*}\in X^{*}}\norm{\vx-\vx^{*}}\right\} \\
 & \leq\max_{\vx\in\redgevec~:~g_{j}(\vx)\leq g_{j}(\vx_{0})}\norm{\vx-\vc}\\
 & \leq\max_{\vx\in\redgevec~:~g_{j}(\vx)\leq g_{j}(\vx_{0})}\norm{\mvar P\vx-\mvar P\vy}\\
 & \leq\max_{\vx\in\redgevec~:~g_{j}(\vx)\leq g_{j}(\vx_{0})}\norm{\mvar P\vec{x}}+\norm{\mvar P\vy}\\
 & \leq2\norm{\vec{y_{0}}}+\norm{\mvar P\vec{x}-\vec{y_{j}}}+\norm{\mvar P\vec{y}-\vec{y_{j}}}\\
 & \leq2\norm{\vec{y_{0}}}+2\norm{\mvar P\vec{x}-\vec{y_{j}}}\\
 & \leq4\left(1+2^{-j}\norm{\mvar P}\right)\text{OPT}+2Kt_{j}.
\end{align*}
Similar to Lemma \ref{lem:lip_constant_of_g}, the Lipschitz constant
$L_{j}$ of $g_{j}$ is $\norm{\mvar P}^{2}/t_{j}$. Hence, Theorem
\ref{thm:gradient_descent} shows that 
\begin{eqnarray*}
g_{j}(\vec{x_{k_{j}}}) & \leq & \min_{\vx}g_{j}(\vx)+\frac{2\cdot L_{j}\cdot R_{j}^{2}}{k_{j}+4}\\
 & \leq & \min_{\vx}\norm{\mvar P\vec{x}-\vec{y_{j}}}+\frac{2\cdot L_{j}\cdot R_{j}^{2}}{k_{j}+4}+Kt_{j}
\end{eqnarray*}
So, we have 
\begin{eqnarray*}
\norm{\mvar P\vx_{k_{j}}-\vec{y_{j}}} & \leq & f_{t_{j}}(\mvar P\vec{x_{k_{j}}}-\vec{y_{j}})\\
 & \leq & \text{OPT}+Kt_{j}+\frac{2\norm{\mvar P}^{2}}{t_{j}(k_{j}+4)}\left(4\left(1+2^{-j}\norm{\mvar P}\right)\text{OPT}+2Kt_{j}\right)^{2}.
\end{eqnarray*}
When $2^{-j}\norm{\mvar P}>1$, we have 
\[
t_{j}=\frac{2^{-(j+2)}\norm{\mvar P}\text{OPT}}{K}\quad\text{and}\quad k_{j}=3200\norm{\mvar P}^{2}K
\]
and hence 
\[
\norm{\vec{y_{j+1}}}=\norm{\mvar P\vx_{k_{j}}-\vec{y_{j}}}\leq\left(1+2^{-j-1}\norm{\mvar P}\right)\text{OPT}.
\]
When $2^{-j}\norm{\mvar P}\leq1$, we have 
\[
t_{j}=\frac{\varepsilon\text{OPT}}{2K}\quad\text{and}\quad k_{j}=\frac{800\norm{\mvar P}^{2}K}{\varepsilon^{2}}
\]
and hence 
\[
\norm{\vec{y_{j+1}}}=\norm{\mvar P\vx_{k_{j}}-\vec{y_{j}}}\leq\left(1+\varepsilon\right)\text{OPT}.
\]
Since $\vec{y}_{\text{last}}$ is $\vec{y}$ plus some vectors in
$L$, $\vec{y}-\vec{y}_{\text{last}}\in L$ and $\norm{\vec{y}-\vec{y}_{\text{last}}-\vy}=\norm{\vec{y}_{\text{last}}}\leq\left(1+\varepsilon\right)\text{OPT}.$

\end{proof}

\subsection{Maximum Concurrent Flow}

\global\long\def\smaxone{\text{smax}L1}

For an arbitrary set of demands $\demands_{i}\in\rvertvec$ with $\sum_{v\in V}\demands_{i}(v)=0$
for $i=1,\cdots,k$, we wish to solve the following\emph{ maximum
concurrent flow} problem 
\[
\max_{\alpha\in R,\flow\in\R^{E}}\alpha\enspace\text{subject to}\enspace\incMatrix^{T}\flow_{i}=\alpha\demands_{i}\enspace\text{and}\enspace\normInf{\mvar U^{-1}\sum_{i=1}^{k}|\flow_{i}|}\leq1.
\]
Similar to \sectionref{sub:MaxFlow_formulation}, it is equivalent
to the problem 
\[
\min_{\circVec\in\R^{E\times[k]}}\normInf{\sum_{i=1}^{k}\left|\vvar{\alpha_{i}}+\left(\mvar Q\vx\right)_{i}\right|}
\]
where $\mvar Q$ is a projection matrix onto the subspace $\{\mvar B^{T}\mvar U\vec{x_{i}}=0\}$, the output maximum concurrent flow is 
\[
\flow_{i}(\vx)=\capacityMatrix(\vvar{\alpha}_{i}+\left(\mvar Q\vx\right)_{i})/\normInf{\sum_{i=1}^{k}\left|\vvar{\alpha_{i}}+\left(\mvar Q\vx\right)_{i}\right|}
\enspace,
\]
and $\mvar U\vec{\alpha_{i}}$ is any flow such that $\incMatrix^{T}\mvar U\vec{\alpha_{i}}=\demands_{i}$.
In order to apply \textbf{NonlinearProjection}, we need to find a
regularized norm and a good projection matrix. Let us define the norm 
\[
\norm{\vx}_{1;\infty}=\max_{e\in E}\sum_{i=1}^{k}|x_{i}(e)|.
\]
The problem is simply $\norm{\vec{\alpha}+\mvar Q\vx}_{1;\infty}$
where $\mvar Q$ is a projection matrix from $\mathbb{R}^{E\times[k]}$
to $\mathbb{R}^{E\times[k]}$ onto some subspace. Since each copy
$\mathbb{R}^{E}$ is same, there is no reason that there is coupling
in $\mvar Q$ between different copies of $\mathbb{R}^{E}$ . In the
next lemma, we formalize this by the fact that any good projection
matrix$\mvar P$ onto the subspace $\{\mvar B^{T}\mvar U\vec{x}=0\}\subset\R^{E}$
extends to a good projection $\mvar Q$ onto the subspace $\{\mvar B^{T}\mvar U\vec{x_{i}}=0\}\subset\R^{E\times[k]}$.
Therefore, we can simply extends the good circulation projection $\mvar P$
by formula $\left(\mvar Q\vx\right)_{i}=\mvar P\vec{x}_{i}$. Thus,
the only last piece needed is a regularized $\norm{\cdot}_{1;\infty}$.
However, it turns out that smoothing via conjugate does not work well
in this case because the dual space of $\norm{\cdot}_{1;\infty}$
involves with $\norm{\cdot}_{\infty}$, which is unfavorable for this
kind of smoothing procedure. It can be shown that there is no such
good regularized $\norm{\cdot}_{1;\infty}$. Therefore, we could not
do $O(m^{1+o(1)}k/\epsilon^2)$ using this approach, however, $O(m^{1+o(1)}k^{2}/\epsilon^2)$
is possible by using a bad regularized $\norm{\cdot}_{1;\infty}$. We believe the dependence of $k$ can be improved to $3/2$ using this approach by suitable using Nesterov algorithm because the $\norm{\cdot}_1$ space caused by the multicommodity is a favorable geometry for accelerated methods.

\begin{lemma} \label{lem:smaxL1_properties}Let $\smaxone_{t}(\vx)=\smax_{t}\left(\sum_{i=1}^{k}\sqrt{\left(x_{i}(e)\right)^{2}+t^{2}}\right)$.
It is a convex continuously differentiable function. The Lipschitz
constant of $\nabla\smaxone_{t}$ is $\frac{2}{t}$ and 
\[
\norm{\vx}_{1;\infty}-t\ln(2m)\leq\smaxone_{t}(\vx)\leq\norm{\vx}_{1;\infty}+kt.
\]
\end{lemma}

\begin{proof} 1) It is clear that $\smaxone_{t}$ is smooth.

2) $\smaxone_{t}$ is convex.

Since $\smax_{t}$ is increasing for positive values and $\sqrt{x^{2}+t^{2}}$
is convex, for any $\vx,\vy\in\R^{E\times[k]}$ and $0\leq t\leq1$,
we have 
\begin{eqnarray*}
\smaxone_{t}(t\vx+(1-t)\vy) & = & \smax_{t}\left(\sum_{i=1}^{k}\sqrt{\left((tx_{i}+(1-t)y_{i})(e)\right)^{2}+t^{2}}\right)\\
 & \leq & \smax_{t}\left(\sum_{i=1}^{k}\left(t\sqrt{\left(x_{i}(e)\right)^{2}+t^{2}}+(1-t)\sqrt{\left(y_{i}(e)\right)^{2}+t^{2}}\right)\right)\\
 & \leq & t\smaxone_{t}(\vx)+(1-t)\smaxone_{t}(\vy).
\end{eqnarray*}

3) The Lipschitz constant of $\nabla\smaxone_{t}$ is $\frac{2}{t}$.

Note that $\smax_{t}$ (not its gradient) has Lipschitz constant $1$ because for any
$\vx,\vy\in\R^{E}$, 
\begin{eqnarray*}
 &  & \left|\smax_{t}(\vx)-\smax_{t}(\vy)\right|\\
 & = & \left|t\ln\left(\frac{\sum_{e\in E}\left(\exp(-\frac{x(e)}{t})+\exp(\frac{x(e)}{t})\right)}{2m}\right)-t\ln\left(\frac{\sum_{e\in E}\left(\exp(-\frac{y(e)}{t})+\exp(\frac{y(e)}{t})\right)}{2m}\right)\right|\\
 & = & t\left|\ln\left(\frac{\sum_{e\in E}\left(\exp(-\frac{x(e)}{t})+\exp(\frac{x(e)}{t})\right)}{\sum_{e\in E}\left(\exp(-\frac{y(e)}{t})+\exp(\frac{y(e)}{t})\right)}\right)\right|\\
 & \leq & t\left|\ln\left(\max_{e\in E}\exp(\frac{|x-y|(e)}{t})\right)\right|\\
 & = & \normInf{\vx-\vy}.
\end{eqnarray*}
Also, by the definition of derivative, for any $\vx,\vy\in\R^{n}$
and $t\in\R$, we have 
\[
\smax_{t}(\vx+t\vy)-\smax_{t}(\vx)=t\innerProduct{\nabla\smax_{t}(\vx)}{\vy}+o(t).
\]
and it implies $\left|\innerProduct{\nabla\smax_{t}(\vx)}{\vec{y}}\right|\leq\normInf{\vy}$
for arbitrary $\vy$ and hence 
\begin{equation}
\norm{\nabla\smax_{t}(\vx)}_{1}\leq1.\label{eq:smax_Lip}
\end{equation}

For notational simplicity, let $s_{1}=\smaxone_{t}$, $s_{2}=\smax_{t}$
and $s_{3}(x)=\sqrt{x^{2}+t^{2}}$. Thus, we have 
\[
s_{1}(\vec{x})=s_{2}\left(\sum_{i=1}^{k}s_{3}(x_{i}(e))\right).
\]
Now, we want to prove 
\[
\norm{\nabla s_{1}(\vx)-\nabla s_{1}(\vy)}_{\infty;1}\leq\frac{2}{t}\norm{\vx-\vy}_{1;\infty}.
\]
Note that 
\[
\frac{\partial s_{1}(\vx)}{\partial x_{i}(e)}=\frac{\partial s_{2}}{\partial e}\left(\sum_{j}s_{3}(x_{j}(e))\right)\frac{ds_{3}}{dx}\left(x_{i}(e)\right).
\]
Hence, we have 
\begin{eqnarray*}
\norm{\nabla s_{1}(x)-\nabla s_{1}(y)}_{\infty;1} & = & \sum_{e}\max_{i}\left|\frac{\partial s_{2}}{\partial e}\left(\sum_{j}s_{3}(x_{j}(e))\right)\frac{ds_{3}}{dx}\left(x_{i}(e)\right)-\frac{\partial s_{2}}{\partial e}\left(\sum_{j}s_{3}(y_{j}(e))\right)\frac{ds_{3}}{dx}\left(y_{i}(e)\right)\right|\\
 & \leq & \sum_{e}\max_{i}\left|\frac{\partial s_{2}}{\partial e}\left(\sum_{j}s_{3}(x_{j}(e))\right)\frac{ds_{3}}{dx}\left(x_{i}(e)\right)-\frac{\partial s_{2}}{\partial e}\left(\sum_{j}s_{3}(x_{j}(e))\right)\frac{ds_{3}}{dx}\left(y_{i}(e)\right)\right|\\
 &  & +\sum_{e}\max_{i}\left|\frac{\partial s_{2}}{\partial e}\left(\sum_{j}s_{3}(x_{j}(e))\right)\frac{ds_{3}}{dx}\left(y_{i}(e)\right)-\frac{\partial s_{2}}{\partial e}\left(\sum_{j}s_{3}(y_{j}(e))\right)\frac{ds_{3}}{dx}\left(y_{i}(e)\right)\right|\\
 & = & \sum_{e}\left|\frac{\partial s_{2}}{\partial e}\left(\sum_{j}s_{3}(x_{j}(e))\right)\right|\max_{i}\left|\frac{ds_{3}}{dx}\left(x_{i}(e)\right)-\frac{ds_{3}}{dx}\left(y_{i}(e)\right)\right|\\
 &  & +\sum_{e}\max_{i}\frac{ds_{3}}{dx}\left(y_{i}(e)\right)\left|\frac{\partial s_{2}}{\partial e}\left(\sum_{j}s_{3}(x_{j}(e))\right)-\frac{\partial s_{2}}{\partial e}\left(\sum_{j}s_{3}(y_{j}(e))\right)\right|.
\end{eqnarray*}

Since $s_{3}$ has $\frac{1}{t}$-Lipschitz gradient, we have 
\[
\left|\frac{ds_{3}}{dx}\left(x\right)-\frac{ds_{3}}{dx}\left(y\right)\right|\leq\frac{1}{t}|x-y|.
\]
By (\ref{eq:smax_Lip}), we have 
\[
\sum_{e}\left|\frac{\partial s_{2}}{\partial e}\left(x(e)\right)\right|\leq1.
\]
Hence, we have 
\begin{eqnarray*}
 &  & \sum_{e}\max_{i}\left|\frac{ds_{3}}{dx}\left(x_{i}(e)\right)-\frac{ds_{3}}{dx}\left(y_{i}(e)\right)\right|\left|\frac{\partial s_{2}}{\partial e}\left(\sum_{i}s_{3}(x_{i}(e))\right)\right|\\
 & \leq & \frac{1}{t}\max_{i,e}|x_{i}(e)-y_{i}(e)|\sum_{e}\left|\frac{\partial s_{3}}{\partial e}\left(\sum_{i}s_{3}(x_{i}(e))\right)\right|\\
 & = & \frac{1}{t}\norm{\vec{x}-\vec{y}}_{1;\infty}.
\end{eqnarray*}

Since $s_{3}$ is $1$-Lipschitz, we have 
\[
\left|\frac{ds_{3}}{dx}\right|\leq1.
\]
Since $s_{2}$ has $\frac{1}{t}$-Lipschitz gradient in $\normInf{\cdot}$,
we have 
\[
\sum_{e}\left|\frac{\partial s_{2}}{\partial e}\left(x\right)-\frac{\partial s_{2}}{\partial e}\left(y\right)\right|\leq\frac{1}{t}\normInf{\vx-\vy}.
\]
Hence, we have 
\begin{eqnarray*}
 &  & \sum_{e}\max_{i}\frac{ds_{3}}{dx}\left(y_{i}(e)\right)\left|\frac{\partial s_{2}}{\partial e}\left(\sum_{j}s_{3}(x_{j}(e))\right)-\frac{\partial s_{2}}{\partial e}\left(\sum_{j}s_{3}(y_{j}(e))\right)\right|\\
 & \leq & \sum_{e}\left|\frac{\partial s_{2}}{\partial e}\left(\sum_{j}s_{3}(x_{j}(e))\right)-\frac{\partial s_{2}}{\partial e}\left(\sum_{j}s_{3}(y_{j}(e))\right)\right|\\
 & \leq & \frac{1}{t}\normInf{\sum_{i}s_{3}(x_{i}(e))-\sum s_{3}(y_{i}(e))}\\
 & \leq & \frac{1}{t}\norm{\sum_{i}\left|x_{i}(e)-y_{i}(e)\right|}\\
 & = & \frac{1}{t}\norm{\vx-\vy}_{1;\infty}
\end{eqnarray*}
Therefore, we have

\[
\norm{\nabla s_{1}(\vx)-\nabla s_{1}(\vy)}_{\infty;1}\leq\frac{2}{t}\norm{\vx-\vy}_{1;\infty}.
\]

4) Using the fact that 
\[
\norm{x(e)}\leq\sum_{i=1}^{k}\sqrt{\left(x_{i}(e)\right)^{2}+t^{2}}\leq\norm{x(e)}_{1}+kt
\]
and $\smax$ is 1-Lipschitz, we have 
\[
\norm{\vx}_{1;\infty}-t\ln(2m)\leq\smaxone_{t}(\vx)\leq\norm{\vx}_{1;\infty}+kt.
\]

\end{proof}

The last thing needed is to check is that the $\#$ operator is easy
to compute.

\begin{lemma} \label{lem:sharp_formula_Linf}In $\norm{\cdot}_{1;\infty}$,
the $\#$ operator is given by an explicit formula 
\[
\left(\dualVec{\vx}\right)_{i}(e)=\begin{cases}
||\vec{x}||_{\infty;1}\mathrm{sign}(x_{i}(e)) & \text{if }i\text{ is the smallest index such that }\min_{j}|x_{j}(e)|=x_{i}(e)\\
0 & \text{otherwises}
\end{cases}.
\]

\end{lemma}

\begin{proof} It can be proved by direct computation.

\end{proof}

Now, all the conditions in the Assumption \ref{ass:nonlinear_projection}
are satisfied. Therefore, Theorem \ref{thm:NonLinearProjection} and
Theorem \ref{thm:master_theorem} gives us the following theorem:

\begin{theorem} \label{thm:approx_max_conc}Given an undirected capacitated graph $G=(V,E,\capacityVec)$
with capacity ratio $U$. Assume $U=\poly(|V|)$. There is an algorithm
finds an $(1-\varepsilon)$ approximate Maximum Concurrent Flow in
time 
\[
O\left(\frac{k^{2}}{\varepsilon^{2}}|E|2^{O\left(\sqrt{\log|V|\log\log |V|}\right)}\right).
\]
\end{theorem}

\begin{proof} Let $\mvar A$ be the oblivious routing algorithm given
by Theorem \ref{thm:master_theorem}. And we have $\rho(\mvar A)\leq2^{O\left(\sqrt{\log|V|\log\log |V|}\right)}$.
Let us define the scaled circulation projection matrix $\mvar P=\mvar I-\mvar U\mvar A\mvar B^{T}\mvar U^{-1}$.
Lemma \ref{lem:obl_rout_to_circulation_projection} shows that $\normInf{\mvar P}\leq1+2^{O\left(\sqrt{\log|V|\log\log |V|}\right)}$. 

Let the multi-commodity circulation projection matrix $\mvar Q:\mathbb{R}^{E\times[k]}\rightarrow\mathbb{R}^{E\times[k]}$
defined by $\left(\mvar Q\vec{x}\right)_{i}=\mvar P\vec{x_{i}}.$
Note that the definition of $\norm{\mvar Q}_{1;\infty}$ is similar
to $\rho(\mvar Q)$. By similar proof as Lemma \ref{lem:equivalence_of_competetivity_and_norm},
we have $\norm{\mvar Q}_{1;\infty}=\norm{\mvar P}_{\infty}$. Hence,
we have $\norm{\mvar Q}_{1;\infty}\leq1+2^{O\left(\sqrt{\log|V|\log\log |V|}\right)}$.
Also, since $\mvar P$ is a projection matrix on the subspace $\{\vec{x}\in\R^{E}:\mvar B^{T}\mvar U\vec{x}=0\}$,
$\mvar Q$ is a projection matrix on the subspace $\{\vec{x}\in\R^{E\times[k]}:\mvar B^{T}\mvar U\vec{x_{i}}=0\}$.

By Lemma \ref{lem:smaxL1_properties}, the function $\smaxone_{t}(\vx)$
is a convex continuously differentiable function such that the Lipschitz
constant of $\nabla\smaxone_{t}$ is $\frac{2}{t}$ and 
\[
\norm{\vx}_{1;\infty}-t\ln(2m)\leq\smaxone_{t}(\vx)\leq\norm{\vx}_{1;\infty}+kt.
\]

Given an arbitrary set of demands $\demands_{i}\in\rvertvec$, we
find a vector $\vec{y}$ such that
\[
\incMatrix^{T}\mvar U\vec{y}=-\demands_{i}.
\]
Then, we use the \textbf{NonlinearProjection} to solve 
\[
\min_{\incMatrix^{T}\mvar U\vx=0}\norm{\vx-\vy}_{1;\infty}
\]
 using a family of functions $\smaxone_{t}(\vx)+t\ln(2n)$ and the
projection matrix $\mvar Q$. Since each iteration involves calculation
of gradients and $\#$ operator, it takes $O(mk)$ each iteration.
And it takes $\tilde{O}\left(\norm{\mvar Q}_{1;\infty}^{2}K/\varepsilon^{2}\right)$
iterations in total where $K=k+\ln(2m)$. In total, it \textbf{NonlinearProjection
}outputs a $(1+\varepsilon)$ approximate minimizer $\vx$ in time
\[
O\left(\frac{k^{2}}{\varepsilon^{2}}|E|2^{O\left(\sqrt{\log|V|\log\log |V|}\right)}\right).
\]
And it gives a $(1-\varepsilon)$ approximate maximum concurrent flow
$\vec{f_{i}}$ by a direct formula.

\end{proof}

\section{Acknowledgements}

We thank Jonah Sherman for agreeing to coordinate submissions and we thank Satish Rao, Jonah Sherman, Daniel Spielman, Shang-Hua Teng. This work was partially supported by NSF awards 0843915 and 1111109, NSF Graduate Research Fellowship (grant no. 1122374) and Hong Kong RGC grant 2150701.

\bibliographystyle{plain}
\bibliography{klos_maxflow}

\appendix

\section{Some Facts about Norm and Functions with Lipschitz Gradient}

In this section, we present some basic fact used in this paper about norm and dual norm. 
Also, we presented some lemmas about convex functions with Lipschitz gradient.
See \cite{bertsekas1999nonlinear,nesterov2003introductory} for comprehensive discussion. 

\subsection{Norms}

\begin{fact} \label{fact:dualzero} 
\[
\vx=0\enspace\Leftrightarrow\dualVec{\vx}=0.
\]
\end{fact}

\begin{proof} If $\vx=0$ then $\forall\vs\neq0$ we have $\innerProduct{\vx}{\vs}-\frac{1}{2}\norm{\vs}^{2}<0$
but $\innerProduct{\vx}{\vx}-\frac{1}{2}\norm{\vx}^{2}=0$. So we
have $\dualVec{\vx}=0$. If $\vx\neq0$ then let $\vs=\frac{\innerProduct{\vx}{\vx}}{\norm{\vx}^{2}}\vx$
with this choice we have $\innerProduct{\vx}{\vs}-\frac{1}{2}\norm{\vs}^{2}=\frac{1}{2}\frac{\innerProduct{\vx}{\vx}^{2}}{\norm{\vx}^{2}}>0$.
However, for $\vs=0$ we have that $\innerProduct{\vx}{\vs}-\frac{1}{2}\norm{\vs}^{2}=0$
therefore we have $\dualVec{\vx}\neq0$. \end{proof}

\begin{fact}\label{fact:dual_norm} 
\[
\forall\vx\in\rn\enspace:\enspace\innerProduct x{\dualVec x}=\norm{\dualVec x}^{2}.
\]
\end{fact}

\begin{proof} If $\vx=0$ then $\dualVec{\vx}=0$ by Claim \ref{fact:dualzero}
and we have the result. Otherwise, again by claim \ref{fact:dualzero}
we know that $\dualVec{\vx}\neq0$ and therefore by the definition
of $\dualVec{\vx}$ we have 
\[
1=\argmax_{c\in\R}\innerProduct x{c\cdot\dualVec x}-\frac{1}{2}\norm{c\cdot\dualVec x}^{2}=\argmax_{c\in\R}c\cdot\innerProduct x{\dualVec x}-\frac{c^{2}}{2}\norm{\dualVec x}^{2}
\]
Setting the derivative of with respect to $c$ to 0 we get that $1=c=\frac{\innerProduct{\vx}{\dualVec{\vx}}}{\norm{\dualVec{\vx}}^{2}}$.
\end{proof}

\begin{fact}  \label{fact:dualnorm_norm} 
\[
\forall\vx\in\rn\enspace:\enspace\normDual{\vx}=\norm{\dualVec{\vx}}.
\]
\end{fact}

\begin{proof} Note that if $\vx=0$ then the claim follows from Claim
\eqref{fact:dualzero} otherwise we have 
\[
\norm{\vx}^{*}=\max_{\norm y\leq1}\innerProduct xy=\max_{\norm y=1}\innerProduct xy\leq\max_{y\in\rn}\frac{\innerProduct xy}{\norm y}
\]
From this it is clear that $\normDual{\vx}\geq\norm{\dualVec{\vx}}$.
To see the other direction consider a $\vy$ that maximizes the above
and let $\vvar z=\frac{\innerProduct{\vx}{\vy}}{\norm{\vy}^{2}}\vy$
\[
\innerProduct{\vx}{\vvar z}-\frac{1}{2}\norm{\vvar z}^{2}\leq\innerProduct{\vx}{\dualVec{\vx}}-\frac{1}{2}\norm{\dualVec{\vx}}^{2}
\]
and therefore 
\[
{\normDual{\vx}}^{2}-\frac{1}{2}{\normDual{\vx}}^{2}\leq\frac{1}{2}\norm{\dualVec{\vx}}^{2}
\]

\end{proof}

\begin{fact}\label{claim:CS_inq}{[}Cauchy Schwarz{]} 
\[
\forall\vx,\vy\in\rn\enspace:\enspace\innerProduct{\vy}{\vx}\leq\normDual{\vy}\norm{\vx}.
\]
\end{fact}

\begin{proof} By the definition of dual norm, for all $\norm{\vx}=1$,
we have $\innerProduct{\vy}{\vx}\leq\normDual{\vy}.$ Hence, it follows
by linearity of both side.\end{proof}

\subsection{Functions with Lipschitz Gradient}

\begin{lemma} \label{lem:lipeq}Let $f$ be a continuously differentiable
convex function. Then, the following are equivalence:

\[
\forall\vx,\vy\in\rn\enspace:\enspace\normDual{\gradient f(\vx)-\gradient f(\vy)}\leq\funLip\cdot\norm{\vx-\vy}
\]
and 
\[
\forall\vx,\vy\in\rn\enspace:\enspace f(\vx)\leq f(\vy)+\innerProduct{\gradient f(\vy)}{\vx-\vy}+\frac{\funLip}{2}\norm{\vx-\vy}^{2}.
\]
For any such $f$ and any $\vx\in\R^{n}$ , we have 
\[
f\left(\vx-\frac{1}{L}\dualVec{\nabla f(\vx)}\right)\leq f(\vx)-\frac{1}{2\funLip}{\normDual{\gradient f(\vx)}}^{2}.
\]
\end{lemma}

\begin{proof} From the first condition, we have 
\begin{eqnarray*}
f(\vy) & = & f(\vx)+\int_{0}^{1}\frac{d}{dt}f(\vx+t(\vy-\vx))dt\\
 & = & f(\vx)+\int_{0}^{1}\innerProduct{\nabla f(\vx+t(\vy-\vx))}{\vy-\vx}dt\\
 & = & f(\vx)+\innerProduct{\nabla f(\vx)}{\vy-\vx}+\int_{0}^{1}\innerProduct{\nabla f(\vx+t(\vy-\vx))-\nabla f(\vx)}{\vy-\vx}dt\\
 & \leq & f(\vx)+\innerProduct{\nabla f(\vx)}{\vy-\vx}+\int_{0}^{1}\normDual{\nabla f(\vx+t(\vy-\vx))-\nabla f(\vx)}\norm{\vy-\vx}dt\\
 & \leq & f(x)+\innerProduct{\nabla f(\vx)}{\vy-\vx}+\int_{0}^{1}Lt\norm{\vy-\vx}^{2}dt\\
 & = & f(x)+\innerProduct{\nabla f(\vx)}{\vy-\vx}+\frac{L}{2}\norm{\vy-\vx}^{2}.
\end{eqnarray*}
Given the second condition. For any $\vx\in\R^{n}$. let $\phi_{\vx}(\vy)=f(\vy)-\innerProduct{\nabla f(\vx)}{\vy}$.
From the convexity of $f$, for any $\vy\in\R^{n}$ 
\[
f(\vy)-f(\vx)\geq\innerProduct{\nabla f(\vx)}{\vy-\vx}.
\]
Hence, $\vx$ is a minimizer of $\phi_{\vx}$. Hence, we have

\begin{align*}
\phi_{\vx}(\vx) & \leq\phi_{\vx}(\vy-\frac{1}{L}\dualVec{\nabla\phi_{\vx}(\vy)})\\
 & \leq\phi_{\vx}(\vy)-\innerProduct{\nabla\phi_{\vx}(\vy)}{\frac{1}{L}\dualVec{\nabla\phi_{\vx}(\vy)}}+\frac{L}{2}\norm{\frac{1}{L}\dualVec{\nabla\phi_{\vx}(\vy)}}^{2}\tag{First part of this lemma}\\
 & =\phi_{\vx}(\vy)-\frac{1}{2L}\norm{\dualVec{\nabla\phi_{\vx}(\vy)}}^{2}\\
 & =\phi_{\vx}(\vy)-\frac{1}{2L}\left(\normDual{\nabla\phi_{\vx}(\vy)}\right)^{2}.
\end{align*}
Hence, 
\[
f(\vy)\geq f(\vx)+\innerProduct{\nabla f(\vx)}{\vy-\vx}+\frac{1}{2L}\left(\normDual{\nabla f(\vy)-\nabla f(\vx)}\right)^{2}.
\]
Adding up this inequality with $\vx$ and $\vy$ interchanged, we
have 
\begin{eqnarray*}
\frac{1}{L}\left(\normDual{\nabla f(\vy)-\nabla f(\vx)}\right)^{2} & \leq & \innerProduct{\nabla f(\vy)-\nabla f(\vx)}{\vy-\vx}\\
 & \leq & \normDual{\gradient f(\vy)-\gradient f(\vx)}\norm{\vy-\vx}.
\end{eqnarray*}
The last inequality follows from similar proof in above for $\phi_{\vx}$.\end{proof}

The next lemma relate the Hessian of function with the Lipschitz parameter
$\funLip$ and this lemma gives us a easy way to compute $L$. \begin{lemma}
\label{lem:lip_formula}Let $f$ be a twice differentiable function
such that for any $\vx,\vy\in\R^{n}$ 
\[
0\leq\vy^{T}\left(\nabla^{2}f(\vx)\right)\vy\leq L||\vy||^{2}.
\]
Then, $f$ is convex and the gradient of $f$ is Lipschitz continuous
with Lipschitz parameter $\funLip$.\end{lemma} \begin{proof} Similarly
to Lemma \ref{lem:lipeq}, we have 
\begin{eqnarray*}
f(\vy) & = & f(\vx)+\innerProduct{\nabla f(\vx)}{\vy-\vx}+\int_{0}^{1}\innerProduct{\nabla f(\vx+t(\vy-\vx))-\nabla f(\vx)}{\vy-\vx}dt\\
 & = & f(\vx)+\innerProduct{\nabla f(\vx)}{\vy-\vx}+\int_{0}^{1}t(\vy-\vx)^{T}\nabla^{2}f(\vx+\theta_{t}(\vy-\vx))(\vy-\vx)dt
\end{eqnarray*}
where the $0\leq\theta_{t}\leq t$ comes from mean value theorem.
By the assumption, we have 
\begin{eqnarray*}
f(\vx)+\innerProduct{\nabla f(\vx)}{\vy-\vx} & \leq & f(\vy)\\
 & \leq & f(\vx)+\innerProduct{\nabla f(\vx)}{\vy-\vx}+\int_{0}^{1}tL\norm{\vy-\vx}^{2}dt\\
 & \leq & f(\vx)+\innerProduct{\nabla f(\vx)}{\vy-\vx}+\frac{L}{2}\norm{\vy-\vx}^{2}.
\end{eqnarray*}
And the conclusion follows from Lemma \ref{lem:lipeq}. \end{proof}

\end{document}